\documentclass[a4paper,11pt]{article}
\pdfoutput=1 

\usepackage{jheppub2} 
\usepackage{bbm,bm,graphicx,mathtools,color,slashed,hyperref}
\usepackage{amssymb,amsmath}
\usepackage{braket}
\usepackage{subfig}

\newcommand{\diff}{\mathrm{d}}

\newcommand{\Diff}{{\mathcal{D}}}

\newcommand{\tr}{\mathrm{tr}}
\newcommand{\im}{\mathrm{i}}

\newcommand{\rme}{\mathrm{e}}

\newcommand{\sfB}{\mathsf{B}}
\newcommand{\ord}{\operatorname{ord}}
\newcommand{\lcm}{\operatorname{lcm}}

\usepackage{amsthm}

\theoremstyle{definition}
\newtheorem{claim}{Claim}
\newtheorem{assumption}{Assumption}
\newtheorem{definition}{Definition}
\newenvironment{definitionprime}{%
  \addtocounter{definition}{-1}% 
  \begin{definition}%
}{%
  \end{definition}%
}

\title{
Monopoles, Center Vortices, Confinement in (3+1)d, and the Lens-Space Twisted Partition Function
}

\author[1]{Yui Hayashi}
\emailAdd{yui.hayashi@yukawa.kyoto-u.ac.jp}
\affiliation[1]{Department of Physics, The University of Tokyo, 7-3-1 Hongo, Bunkyo-ku, Tokyo 113-0033, Japan}

\author[2]{and Yuya Tanizaki}
\emailAdd{yuya.tanizaki@yukawa.kyoto-u.ac.jp}
\affiliation[2]{Yukawa Institute for Theoretical Physics, Kyoto University, Kyoto 606-8502, Japan}

\preprint{YITP-26-63}

\abstract{
We propose the gauge-invariant criteria of center-vortex condensation and monopole condensation using the $\mathbb{Z}_N^{[1]}$-symmetry twisted partition functions: 
The torus twisted partition function characterizes the center-vortex condensation, and the lens-space twisted partition function characterizes the monopole condensation. 
To justify our proposal, we study how these twisted partition functions behave in the adjoint Higgs phase and show that their leading nontrivial contributions come from the center vortex and monopole, respectively. 
Using the techniques of topological field theories, we uncover the relation between the center-vortex and monopole condensations, and in particular, we prove that the gapped phase with the center-vortex condensation necessarily shows the monopole condensation, too. 
We then study a center-vortex model with monopoles as an illustrative example, and the higher-charge monopole condensation gives an example of the symmetry fractionalization, which goes beyond the conventional Wilson-'t~Hooft classification. 
}

\begin{document}

\maketitle

\section{Introduction}
\label{sec:introduction}

In the confinement phase, the color-electric flux emitted from the test quark is collimated into a string, which produces the linear inter-quark potential, and one can measure it by observing the area law for the Wilson loops. 
Excitations of magnetic objects are the key for the intuitive explanation of confinement, and the two most standard scenarios are the monopole condensation and the center-vortex proliferation~\cite{Nambu:1974zg, Mandelstam:1974pi,  Polyakov:1975rs, tHooft:1977nqb, Cornwall:1979hz, Nielsen:1979xu, Ambjorn:1980ms}.

The monopole condensation scenario explains confinement via the electromagnetic dual of superconductivity: As a consequence of the condensation of magnetically charged particles, the magnetic field is screened, and the electric field is expelled, leading to the formation of the electric flux tube between two test quarks. 
The picture of Abelian dominance and dual Meissner effect~\cite{Ezawa:1982bf, Suzuki:1988yq, Suganuma:1993ps, Kondo:1997pc} has been tested in the numerical simulation~\cite{Kronfeld:1987vd, Suzuki:1989gp}, and the observation tells the monopole components of the lattice gauge fields give a large contribution for the confining string in a suitable gauge choice, e.g. the maximally Abelian gauge.

The center-vortex scenario gives a complementary picture for the confinement phase. 
The $SU(N)$ Yang-Mills theory (possibly with adjoint matters) enjoys the $\mathbb{Z}_N$  $1$-form center symmetry, $\mathbb{Z}_N^{[1]}$, and the Wilson loop gives an order parameter for this global symmetry. 
In the deconfined phase, where the $1$-form symmetry is spontaneously broken, there exists a vortex soliton around which the Wilson loop gets the nontrivial $\mathbb{Z}_N$ Aharonov–Bohm phase. 
This is called the center vortex, which is a counterpart of the domain-wall excitation for the case of the usual discrete symmetry breaking. 
Then, the proliferation of center vortices would restore the spontaneously broken $1$-form symmetry and explain quark confinement. 
Analogous to the Abelian dominance in the maximal Abelian gauge, the lattice numerical simulation observes the dominance of the center vortex configurations for the confining force in a maximal center gauge \cite{DelDebbio:1996lih, Langfeld:1998cz, Kovacs:1998xm, Engelhardt:1999fd, deForcrand:1999our}.

Numerical success of both pictures suggests the monopole and center-vortex condensation scenarios are related for the case of $(3+1)$d Yang-Mills theory. 
Indeed, one can naturally incorporate the monopoles into the center-vortex ensembles: Monopole worldlines serve as the junctions in the network of center-vortex worldsheets~\cite{DelDebbio:1997ke, Ambjorn:1999ym, deForcrand:2000pg}. 
Local instanton fluctuations can be produced only with such non-orientable vortex worldsheets in the center-vortex picture~\cite{Engelhardt:1999xw, Reinhardt:2001kf, Cornwall:1999xw}, which highlights the importance of monopoles in the vortex scenario. 
In a compactified geometry, such a center vortex with the monopole junction can be realized as the self-dual solution of the classical Yang-Mills equation~\cite{Hayashi:2024yjc} (see also~\cite{GarciaPerez:1989gt, GarciaPerez:1992fj, Ford:2002pa, Ford:2003vi, Itou:2018wkm, DasilvaGolan:2022jlm, Wandler:2024hsq} for related fractional instantons), and it plays the essential role to explicitly interpolate the semiclassical description of confinement via center vortices on $\mathbb{R}^2\times T^2_{\mathrm{twisted}}$~\cite{Tanizaki:2022ngt, Tanizaki:2022plm, Hayashi:2023wwi, Hayashi:2024gxv, Hayashi:2024qkm, Hayashi:2024psa, Guvendik:2024umd, Hayashi:2025doq} (see also \cite{Yamazaki:2017ulc, Cox:2021vsa}) and the one via monopole instantons on $\mathbb{R}^3\times S^1$~\cite{Unsal:2007vu, Unsal:2007jx, Unsal:2008ch, Shifman:2008ja, Poppitz:2012nz, Poppitz:2012sw, Davies:2000nw, Davies:1999uw}.

Recently, the connection between the monopoles and center vortices has been revisited in the context of the $4$d lattice $\mathbb{Z}_N$ gauge theories~\cite{Nguyen:2024ikq, Giansiracusa:2025hfj}. 
The authors of Refs.~\cite{Nguyen:2024ikq, Giansiracusa:2025hfj} found the gapless phase with an emergent massless photon if the vacuum configurations are filled with center vortices without monopole junctions. 
In the confinement phase, on the other hand, both monopoles and center vortices are abundant in the vacuum, which strongly suggests the importance of monopoles for the center-vortex scenario of the confinement mechanism. 
One of our motivation in this paper is to give a formal argument for their numerical observation in a way that is extendable to non-Abelian gauge theories. 

One of the obstacles for this purpose is the lack of gauge-invariant definition for monopoles and center vortices in the $SU(N)$ Yang-Mills theory. 
Therefore, we first give a ``proposal'' to characterize the magnetic-charge condensation and the center-vortex condensation in terms of the symmetry-twisted partition function with the $\mathbb{Z}_N$ $1$-form symmetry: The magnetic-charge condensation is defined by using the partition function on the lens space, $L(qN;1)\times S^1$, and the center-vortex condensation is defined by using the partition function on the torus, $T^4$.  
By construction, our definition is gauge invariant, and we shall explain why this gives a reasonable definition by calculating those partition functions in adjoint Higgs phases. 
Let us note that we do not give gauge-invariant definition of monopoles and center vortices themselves, but we propose the criteria for condensations. 

Once we get a gauge-invariant definition, we can discuss their relation by using the technique of the topological field theory under the presence of the mass gap. 
We then prove that the center-vortex condensation requires the magnetic-charge condensation, but we also find that the condensed magnetic charge does not have to be minimal ones. 
It turns out that the lens-space twisted partition function provides us detailed information of the condensed magnetic charge, which has been missed in the torus partition function and also the conventional Wilson-'t Hooft classification. 
We also give a concrete demonstration of these abstract observations by considering the center-vortex percolation model with monopole worldlines. 

\section{Criteria of center-vortex/monopole condensation via \texorpdfstring{$\mathbb{Z}_N^{[1]}$}{ZN[1]} symmetry}
\label{sec:Definitions}

$SU(N)$ Yang-Mills theory (with adjoint matters) enjoys the $\mathbb{Z}_N$ $1$-form symmetry, which is denoted by $\mathbb{Z}_N^{[1]}$: There exists the codimension-$2$ topological operator with the $\mathbb{Z}_N$ fusion rule, and it acts on the Wilson loops by the $\mathbb{Z}_N$ phase rotation~\cite{Gaiotto:2014kfa, Kapustin:2014gua, Sharpe:2015mja}. 
For local field theories with the $\mathbb{Z}_N^{[1]}$ symmetry, we can introduce its background gauge field $B$ described by the $\mathbb{Z}_N$ $2$-form gauge field and compute its partition function, $Z_{M_4}[B]$, on the general closed $4$-manifolds $M_4$. 
This operation is equivalent to the insertion of the codimension-$2$ topological operators/defects on the Poincar\'e dual codim-$2$ cycle of $B$.

We shall propose the gauge-invariant criteria for the monopole and center-vortex condensation using the twisted partition functions in Sec.~\ref{sec:detect_condensation} after the brief explanation of the properties of $(3+1)$d gapped phases in view of topological field theories in Sec.~\ref{sec:TQFT_basics}. 
We then explain in Sec.~\ref{sec:justification} why our proposals are physically sensible by computing the symmetry-twisted partition functions in the Higgs phase, where the dynamical monopoles/center vortices are well defined. 
Section~\ref{sec:main_result} is one of the main results of this paper, which claims that gapped phases with the center-vortex condensation necessarily have the monopole condensation, too. 

% In this section, we propose plausible gauge-invariant criteria for the monopole and center-vortex condensations in terms of the $\mathbb{Z}_N$ $1$-form symmetry.
% %A difficulty in studying confinement scenarios of $4$d gauge theories is the lack of gauge-invariant formulations for most proposed mechanisms. To establish rigorous results, clarifying these definitions is essential. 
% Because the $\mathbb{Z}_N^{[1]}$ defect intuitively acts as a probe center vortex, it naturally enables the definition of center-vortex condensation. 
% %Conversely, formulating monopole condensation is highly non-trivial, as it conventionally relies on abelianization. In this paper, we propose a definition of ``monopole condensation'' strictly in terms of $\mathbb{Z}_N^{[1]}$ defects. 
% For monopoles, a mathematical trick is necessary. 
% By introducing a specialized configuration of these $\mathbb{Z}_N^{[1]}$ defects, we can make a setup in which the presence of monopoles (or more accurately, sources of magnetic flux) is inevitably induced if abelianization occurs. Consequently, a criterion of monopole condensation can be given by the non-suppression of this twisted partition function.

%\blue{When the system is gapped, the theory will flow to a TQFT with $\mathbb{Z}_N$ $1$-form symmetry , and the classification of such theories is known as the Wilson-'t~Hooft classification. (Under proper assumptions,) the (locally-unique-gapped) TQFTs can be described by the $\mathbb{Z}_N^{[1]} \rightarrow \mathbb{Z}_n^{[1]}$ spontaneous breakdown with $\mathbb{Z}_n^{[1]}$ SPT stacking.}

\subsection{\texorpdfstring{$(3+1)$}{(3+1)}d  gapped phases, local TQFTs, and remote detectability}
\label{sec:TQFT_basics}

Let us consider a $(3+1)$-dimensional local and unitary relativistic quantum field theory (QFT) with the $\mathbb{Z}_N^{[1]}$ symmetry, and we here assume that it has a nonzero mass gap, i.e. all the local correlation functions decay exponentially, and also that the ground state has no local degeneracy. 
For this setup, the low-energy effective theory is described by a topological quantum field theory (TQFT), which also satisfies locality and unitarity. 
Let us write $Z^{\mathrm{(TQFT)}}_{M_4}[B]$ for the partition function of the low-energy TQFT with the background gauge field $B$, and then the full QFT partition function should behave as 
\begin{equation}
    Z_{M_4}[B]\xrightarrow{\mathrm{vol.}\to\infty} 
    \exp\left(-\int_{M_4}\sqrt{g_{M_4}}\diff^4 x\Bigl(\Lambda+\frac{1}{2\kappa} R_{M_4}+\cdots\Bigr)\right)
    Z^{(\mathrm{TQFT})}_{M_4}[B]  
\end{equation}
in the infinite volume limit. 
Here, the exponential factor on the right-hand-side describes the gravitational local counterterm; $g_{M_4}$ is the metric of $M_4$, $R_{M_4}$ is its Ricci scalar, and $\Lambda$, $\kappa$, etc. are non-universal (or regularization-dependent) constants, which are commonly called cosmological constant, Newton constant, etc., respectively.

Thanks to the locality, the QFT partition functions on various closed manifolds are not totally independent but they are related by the ``cutting and gluing principle''~\cite{Witten:1988hf, Atiyah:1989vu, Baez:1995xq, Lurie:2009keu, Kapustin:2010ta, Freed:2016rqq}. 
Since our argument strongly relies on it, let us state it here clearly:
\begin{assumption}[Locality]
Consider the closed $4$-manifold $M_4$ that is composed by gluing two open $4$-manifolds $N^{(1)}_4$ and $N^{(2)}_4$ with the common $3$-boundary $Y_3=\partial N^{(1)}_4=\partial N^{(2)}_4$; 
\begin{align}
    M_4=\overline{N}^{(1)}_4\cup_{Y_3} N^{(2)}_4, 
\end{align} 
where $\overline{N}^{(1)}_4$ is the orientation reversal of $N^{(1)}_4$. 
Then, the QFT partition function on $M_4$ can be decomposed as 
\begin{equation}
    Z_{M_4}[B]=\bra{Z_{N^{(1)}_4}[B|_{N^{(1)}_4}]} \hat{U}(B_{Y_3\perp}) \ket{Z_{N^{(2)}_4}[B|_{N^{(2)}_4}]}. 
\end{equation}
Here, $\ket{Z_{N^{(1,2)}_4}[B]}$ is the state obtained by the path integral over $N^{(1,2)}_4$ submanifolds, which belongs to the (twisted) Hilbert space on $Y_3$ as 
\begin{align}
    \ket{Z_{N^{(1,2)}_4}[B]}\in \mathcal{H}(Y_3, B_{Y_3 ||}\text{-twisted}). 
\end{align}
The twisted boundary condition is specified by the parallel component, $B_{Y_3||}$, of the background gauge field $B$ along $Y_3$. 
When $B$ has the transverse component against $Y_3$, it introduces the unitary operator $\hat{U}(B_{Y_3\perp})$ for the $1$-form symmetry. \qed
\end{assumption}

This locality assumption is valid not only in the low-energy TQFT but also in the full ultraviolet description of the QFT with suitable regularization. 
For our discussion, we especially focus on the $4$-manifold with the form $M_4=Y_3\times S^1$, and then the TQFT partition function should behave as 
\begin{align}
    Z^{(\mathrm{TQFT})}_{Y_3\times S^1}[B]=\tr_{\mathcal{H}_0(Y_3,B_{Y_3||}\text{-twisted})}[\hat{U}(B_{Y_3\perp})], 
    \label{eq:TQFT_trace}
\end{align}
and the normalization of the TQFT partition function is strictly constrained by the topologicalness and also the cut-and-glue rule. 
Here, we put a subscript for the Hilbert space as $\mathcal{H}_0(Y_3,B_{Y_3||}\text{-twisted})$ to clarify that it is the Hilbert space of the low-energy TQFT, which only collects the zero-energy state (after subtracting gravitational counterterms) in the full QFT. 
For the low-energy TQFTs, we also put the following physical assumptions, called finiteness and remote detectability: 

\begin{assumption}[Finiteness of TQFT]
\label{assumption:finiteness}
    On any closed $3$-manifolds, the (twisted) Hilbert spaces $\mathcal{H}_0(Y_3,B\text{-twisted})$ of the TQFT are finite dimensional.  \qed
\end{assumption}
\noindent
The finiteness assumption implies the set of deconfined line operators are finitely generated, 
and it ensures the trace expression~\eqref{eq:TQFT_trace} well-defined. 

\begin{assumption}[Remote detectability of TQFT \cite{Lan:2018vjb, Lan:2018bui}]
    $(3+1)$d TQFTs have the line and surface operators. 
    If any non-identity lines have nontrivial commutation with some surface operators and vice versa, the TQFT is said to satisfy the remote detectability. \qed
\end{assumption}
\noindent
The (3+1)d TQFTs have both particle-like and string-like deconfined excitations, whose worldlines and worldsheets are described by the line and surface operators, respectively. 
The remote detectability, introduced in Refs.~\cite{Lan:2018vjb, Lan:2018bui}, claims those deconfined excitations should be detectable via the Aharonov--Bohm-type phase by performing the linking experiments using the complementary objects. 
We note that the locality, finiteness, and unitarity are standard assumptions of TQFTs including the mathematics literature, but the remote detectability is not always built in and we here put it as the extra physical requirement. 
For the (2+1)d TQFTs, the counterpart is known as the braiding nondegeneracy~\cite{Kitaev:2006lla}, and the remote detectability is its generalization for TQFTs in other dimensions. 
As shown physically in Refs.~\cite{Lan:2018vjb, Lan:2018bui} and also mathematically in Ref.~\cite{Johnson-Freyd:2020usu}, the $(3+1)$d TQFTs with remote detectability are constrained very strongly and have the rigid structure: 
\begin{itemize}
\item (3+1)d local, unitary TQFTs with remote detectability are described by discrete $2$-group gauge theories~\cite{Lan:2018vjb, Lan:2018bui,Johnson-Freyd:2020usu}. 
\end{itemize}
This ensures that the untwisted partition function on $Y_3\times S^1$ is a strictly positive integer,
\begin{align}
    Z^{(\mathrm{TQFT})}_{Y_3\times S^1}[0]=\tr_{\mathcal{H}_0(Y_3)}[\bm{1}]=\mathrm{dim}\,\mathcal{H}_0(Y_3) \in \mathbb{Z}_{>0}. 
\end{align}
The standard TQFT axiom assumes $\mathrm{dim}\, \mathcal{H}_0(S^3)=1$, but the Hilbert space can be $\{0\}$ for other closed $3$-manifolds in general TQFTs. 
When it is described by a (higher-)gauge theory, the trivial-connection state exists on any closed spatial manifolds, and thus the untwisted partition function on $Y_3\times S^1$ does not vanish. 
Thus, the division by zero does not appear when computing the ratios of the twisted and untwisted partition functions.

Let us consider the following ratios of twisted partition functions to obtain the low-energy TQFT information out of the full QFT partition functions by completely eliminating the gravitational counterterms: The first one is the Gu-Wen ratio~\cite{Gu:2009dr},
\begin{align}
    \frac{\left(Z_{Y_3\times S^1_{(L)}}[0]\right)^2}{Z_{Y_3\times S^1_{(2L)}}[0]}\xrightarrow{\mathrm{vol.}\to \infty} \mathrm{dim}\, \mathcal{H}_0(Y_3), 
    \label{eq:GuWen_ratio}
\end{align}
which counts the number of the TQFT zero-energy states on $Y_3$, and the second one extracts the TQFT twisted partition function~\cite{Maeda:2025ycr},
\begin{align}
    \frac{Z_{Y_3\times S^1}[B]}{Z_{Y_3\times S^1}[0]}\xrightarrow{\mathrm{vol.}\to\infty} \frac{1}{\dim \mathcal{H}_0(Y_3)} Z^{(\mathrm{TQFT})}_{Y_3\times S^1}[B]. 
    \label{eq:TwsitedZ_ratio}
\end{align}
For the infinite-volume limit, both sizes of $Y_3$ and $S^1$ should be sufficiently large compared with the mass gap.

\subsection{Gauge-invariant criteria for center-vortex/monopole condensation}
\label{sec:detect_condensation}

Let us propose the gauge-invariant criteria for the center-vortex condensation and also for the monopole condensation using the $\mathbb{Z}_N^{[1]}$-twisted partition functions. 
In this section, we describe the definitions and we shall give its justification in the next subsection~\ref{sec:justification}. 

\subsubsection{Weak and strong criteria of center-vortex condensation via \texorpdfstring{$Z_{T^4}[B]$}{Z{T4}[B]}}

We first define the center-vortex condensation by the infinite-volume behavior of the $\mathbb{Z}_N^{[1]}$-twisted partition function on the $4$-torus, $T^4$. We here give the two versions of the definition, which we call the weak and strong versions of the center-vortex condensation: 

\begin{definition}[Center-vortex condensation; weak version]
\label{def:center-vortex-condensation}
We consider the $4$-torus $T^4$ with introducing the minimal nontrivial background $2$-form gauge field $B$ along the $3$-$4$ cycle: 
\begin{align}
    T^4= S^1\times S^1\times \underbrace{S^1\times S^1}_{\int\frac{2\pi}{N}B=\frac{2\pi}{N}}, 
\end{align}
i.e. we may specifically take $B=\delta(T^2_{(12)})=\delta(x_3,x_4)\diff x_3\wedge \diff x_4$. 
If the following condition is satisfied,
\begin{align}
    \frac{Z_{T^4}[B]}{Z_{T^4}[0]} \xrightarrow{\mathrm{vol.\to \infty}} \text{const.} \neq 0,
    \label{eq:centervortex_condensationT4}
\end{align}
then we say that center-vortex condensation in the weak version (or, weak center-vortex condensation) occurs.
\qed
\end{definition}

\begin{definitionprime}
    [Center-vortex condensation; strong version]
\label{def:center-vortex-condensation-strong}
Let us consider the $4$-torus with the same background gauge field $\frac{2\pi}{N}B=\frac{2\pi}{N}\delta(x_3,x_4)\diff x_3\wedge \diff x_4$, and if the following condition is satisfied,
\begin{align}
    \frac{Z_{T^4}[B]}{Z_{T^4}[0]} \xrightarrow{\mathrm{vol.}\to \infty} 1,
    \label{eq:centervortex_condensationT4_strong}
\end{align}
then we say that center-vortex condensation in the strong version (or, strong center-vortex condensation) occurs.
\qed
\end{definitionprime}

\noindent
Let us interpret these definitions of the center-vortex condensation from the operator formalism. When we take the $1$st $S^1$-cycle as the imaginary-time direction, this ratio of the twisted partition function can be interpreted as 
\begin{align}
    \frac{Z_{T^4}[B]}{Z_{T^4}[0]}\xrightarrow{\mathrm{vol.}\to\infty} \frac{\tr_{\mathcal{H}_0(T^3,B\text{-twisted})}[\bm{1}]}{\dim \mathcal{H}_0(T^3)} =\frac{\dim \mathcal{H}_0(T^3, B\text{-twisted})}{\dim \mathcal{H}_0(T^3)},  
\end{align}
and the ratio converges to a non-negative rational number in the infinite-volume limit when the non-zero mass gap is assumed. 
The center-vortex condensation in the weak version requires the presence of the non-zero states in the $\mathbb{Z}_N^{[1]}$-twisted Hilbert space, and the strong version requires that the twisted Hilbert space should have the identical dimension with the untwisted one. 
Instead, if we regard the $4$th $S^1$-cycle as the imaginary-time direction, the operator-formalism interpretation of the ratio becomes 
\begin{align}
    \frac{Z_{T^4}[B]}{Z_{T^4}[0]}\xrightarrow{\mathrm{vol.}\to\infty} \frac{\tr_{\mathcal{H}_0(T^3)}[\hat{U}(B_{T^3\perp})]}{\dim \mathcal{H}_0(T^3)} , 
\end{align}
where the symmetry generator $\hat{U}(B_{T^3\perp})$ is defined over $T^2_{(12)}$. 
In particular, the strong version of the center-vortex condensation requires that the $\mathbb{Z}_N^{[1]}$-symmetry generator $\hat{U}(B_{T^3\perp})$ acts trivially on $\mathcal{H}_0(T^3)$, and thus 
\begin{align}
    &\text{Center-vortex condensation in the strong version} \notag\\
    &\Rightarrow 
    \text{All the $\mathbb{Z}_N^{[1]}$-charged lines show the area law (i.e. Unbroken $\mathbb{Z}_N^{[1]}$ symmetry).} \notag 
\end{align}
This can be shown by contradiction: If there exists a deconfined line having the nontrivial $\mathbb{Z}_N$ commutation relation with the symmetry generator, we consider the line operator along $S^1\subset T^3$ and act it on some of the ground states, which creates a $\mathbb{Z}_N$ charged state,\footnote{Here, we secretly use the remote detectability assumption to ensure this state is a nonzero vector at least for some states on $\mathcal{H}_0(T^3)$ if the line operator is nonzero.} but such states should be absent. 
We note that the converse does not necessarily hold: Let us consider the $\mathbb{Z}_N\times \mathbb{Z}_N$ gauge theory described by the $1$-form gauge fields $a_1,a_2$, and consider the magnetic coupling, $\frac{2\pi i}{N}\int a_1\cup a_2\cup B$. We can explicitly compute that $Z_{T^4}[B]/Z_{T^4}[0]=1/N^2$, while no $\mathbb{Z}_N^{[1]}$-charged lines are deconfined. 
This example satisfies the center-vortex condensation in the weak version while the strong version is violated. 

When assuming the gapped phase with the unbroken $\mathbb{Z}_N^{[1]}$ symmetry, we can show the weak version of the center-vortex condensation as follows. 
Let us prove the contraposition, ``If $\dim \mathcal{H}_0(T^3,\text{twisted})=0$, then the $\mathbb{Z}_N^{[1]}$-charged lines are deconfined''. 
Now, let us consider the 't~Hooft loop operator $H(C,\Sigma)$, which is defined by introducing $B=\delta(\Sigma)$ with the nontrivial boundary, $\partial \Sigma=C\not= \emptyset$. 
Using the cut-and-glue principle, we can think of the 't~Hooft operator as an inner product between the twisted and untwisted Hilbert spaces. 
Since $\mathcal{H}_0(T^3,\mathrm{twisted})=\{0\}$, such an expectation value always vanishes for any kinds of the dressing of the loop operators, i.e. all the dyonic operators with nonzero magnetic charges show the area law. 
Using the Wilson-'t~Hooft classification~\cite{tHooft:1979rtg, Nguyen:2023fun}\footnote{The Wilson-'t Hooft classification referred to here indicates the classification of possible long-distance behaviors of dyonic loops, which is valid. 
Conventionally, it is often supposed that a gapped phase is described by BF theory (including the trivial case) electrically coupled to the $\mathbb{Z}_N^{[1]}$ background $B$ with possible SPT stacking; however, we will see that there are exceptions to this classification due to symmetry fractionalization in Section \ref{sec:highercharge}. In this paper, we refer to the latter as the ``conventional Wilson-'t Hooft classification.''}, this is the signature of the deconfinement for the electric charged lines, which completes the proof. 
Thus, we find 
\begin{align}
    \text{Strong center-vortex condensation} 
    &\Rightarrow 
    \text{Unbroken $\mathbb{Z}_N^{[1]}$ symmetry} \notag\\
    &\Rightarrow
    \text{Weak center-vortex condensation}. 
\end{align}

We note that the strong center-vortex condensation implies ``confinement'' in the usual sense, but the weak center-vortex condensation allows the presence of the deconfined $\mathbb{Z}_N^{[1]}$-charged lines and does not necessarily imply ``confinement'': 
For example, as will be noted in the discussion (Section \ref{sec:discussion}), the weak center-vortex condensation criterion is satisfied by the non-abelian discrete gauge theory with the Heisenberg group $H_N$. 
Such exceptions may be ruled out by requiring the center-vortex condensation criterion on $S^2\times S^2$, which turns out to be completely equivalent to the unbroken $\mathbb{Z}_N^{[1]}$ symmetry with nonzero mass gap.

\subsubsection{Weak and strong criteria of monopole condensation via \texorpdfstring{$Z_{L(qN;1)\times S^1}[B]$}{Z{L(qN;1)xS1}[B]}}

Next, we define the monopole condensation using the $\mathbb{Z}_N^{[1]}$-twisted partition function on the lens space, $L(qN;1)\times S^1$ (For details of the lens space, see Appendix~\ref{app:Lens_space}): 

\begin{definition}[{charge-$[q]_{Nq\mathbb{Z}}$ monopole condensation; weak version}]
\label{def:lens-space-criterion}
Let us take a positive integer $q\in \mathbb{Z}_{>0}$. 
We consider the product space of the lens space $L(qN;1)\simeq S^3/\mathbb{Z}_{qN}$ and $S^1$, i.e. $M_4=L(qN;1)\times S^1$. 
As $\pi_1(L(qN;1))=H_1(L(qN;1))=\mathbb{Z}_{qN}$, let us call its generating $1$-cycle as $\gamma\subset L(qN;1)$. 
The $\mathbb{Z}_N^{[1]}$ background gauge field $B$ takes values in $H^2(L(qN;1)\times S^1;\mathbb{Z}_N)=\mathbb{Z}_N^{\oplus 2}$, and we choose the one as 
\begin{align}
    \int_{\gamma\times S^1} \frac{2\pi}{N}B=\frac{2\pi}{N},  
\end{align}
and we can write it as $B = A_{\mathbb{Z}_N} \cup \delta(x_4)\diff x_4$ using the generator $A_{\mathbb{Z}_N}$ of $H^1(L(qN;1), \mathbb{Z}_N) \simeq \mathbb{Z}_N$. 
Its Bockstein homomorphism $\beta_{\mathbb{Z}}$ associated with $0\to \mathbb{Z}\to \mathbb{Z}\to \mathbb{Z}_{N}\to 0$ becomes nontrivial, 
\begin{align}
    \beta_{\mathbb{Z}}(B)=q\, \delta(\gamma)\in \mathrm{Tor}\, H^3(L(qN;1)\times S^1)\simeq \mathbb{Z}_{qN}.
\end{align} 
If the following condition is satisfied with the above choice of $B$, 
\begin{equation}
    \frac{Z_{L(qN;1) \times S^1}[B]}{Z_{L(qN;1) \times S^1}[0]} \xrightarrow{\text{vol.} \to \infty} \text{const.} \neq 0, 
    \label{eq:Lens-criterion}
\end{equation}
then we say that the charge-$[q]_{Nq\mathbb{Z}}$ monopole condensation in the weak version occurs. 
When there exists, at least one, $q>0$ that satisfies the criterion, we simply say the monopole condensation in the weak version (or, weak monopole condensation) occurs. 
\qed
\end{definition}

\begin{definitionprime}[{charge-$[q]_{Nq\mathbb{Z}}$ monopole condensation; strong version}]
\label{def:lens-space-criterion-strong}
    We consider the same setup as above: We take $M_4=L(qN;1)\times S^1$ with introducing the $\mathbb{Z}_N^{[1]}$ background gauge field $B=A_{\mathbb{Z}}\cup \delta(x_4)\diff x_4$. If the following condition is satisfied, 
    \begin{align}
        \frac{Z_{L(qN;1) \times S^1}[B]}{Z_{L(qN;1) \times S^1}[0]} \xrightarrow{\text{vol.} \to \infty} 1, 
    \label{eq:Lens-criterion-strong}
    \end{align}
    then we say that the charge-$[q]_{Nq\mathbb{Z}}$ monopole condensation in the strong version occurs. 
When there exists, at least one, $q>0$ that satisfies the criterion, we simply say the monopole condensation in the strong version (or, strong monopole condensation) occurs. 
\end{definitionprime}

\noindent
In the operator-formalism interpretation by taking $S^1$ as the imaginary time, the ratio in~\eqref{eq:Lens-criterion} becomes 
\begin{align}
    \frac{Z_{L(qN;1) \times S^1}[B]}{Z_{L(qN;1) \times S^1}[0]} \xrightarrow{\text{vol.} \to \infty} \frac{\tr_{\mathcal{H}_0(L(qN;1))}[\hat{U}(B_{L(qN;1)\perp})]}{\dim \mathcal{H}_0(L(qN;1))}, 
    \label{eq:lens_partition_function_trace}
\end{align}
where the symmetry generator $\hat{U}(B_{L(qN;1)\perp})$ is defined along the $2$-cycle generator for the $2$-nd $\mathbb{Z}_N$-valued homology $H_2(L(qN;1);\mathbb{Z}_N)=\mathrm{Tor}(H_1(L(qN;1)),\mathbb{Z}_N)=\mathbb{Z}_N$. 

As another remark, let us explain about our terminology, ``charge-$[q]_{Nq\mathbb{Z}}$ monopole condensation''. 
As we will see in the following section~\ref{sec:justification}, if the magnetic charge lattice of the condensed particles can screen at least one of the magnetic charges belonging $q+Nq\mathbb{Z}$, the criterion (Def.~\ref{def:lens-space-criterion} or~\ref{def:lens-space-criterion-strong}) is going to be satisfied. 
Thus, we cannot uniquely specify which monopole condensation causes the confinement, and the condensed magnetic charge could be a large one of the form $q+Nqk$ with some $k\gg 1$.  
We also note that the criterion does not distinguish if the condensed particles are monopoles or dyons as long as they create the corresponding magnetic flux.

By repeating the similar discussion for the case of the center-vortex condensation, we can show 
\begin{align}
    \text{Strong monopole condensation} 
    &\Rightarrow 
    \text{Unbroken $\mathbb{Z}_N^{[1]}$ symmetry} \notag\\
    &\Rightarrow
    \text{Weak monopole condensation}. 
\end{align}
Thus, we are interested in the relative strength between the strong center-vortex and monopole condensations and the one between the weak center-vortex and monopole condensations, which shall be discussed in Sec.~\ref{sec:main_result}.

\subsection{Justifications for the proposals: Consideration from the Higgs vacuum}
\label{sec:justification}

For $SU(N)$ gauge theories, it is not generically possible (or, at least, unknown how) to give a gauge-invariant definition of the dynamical center vortex and the dynamical monopoles, which typically assume the Abelianization of the gauge group. 
Despite this difficulty, we have proposed the gauge-invariant definition for the center-vortex condensation (Def.~\ref{def:center-vortex-condensation}) and the monopole condensation (Def.~\ref{def:lens-space-criterion}).
In this section, we shall explain why they are physically reasonable.

The key idea behind these definitions is that the free energy of some soliton (effective tension or effective mass) can be measured by a proper symmetry-twisted partition function.
From this perspective, we expect the following in the present case.
The center-vortex tension can be naturally extracted via the twisted partition function:
\begin{align}
    \frac{Z_{T^4}[B]}{Z_{T^4}[0]} &\sim \exp \big(-L_1 L_2 \times (\text{effective center-vortex tension})\big).
    \label{eq:vortex_motivation}
\end{align}
On the other hand, the lens-space setup is engineered to enforce the appearance of the monopole if the Abelianization occurs:
\begin{align}
    \frac{Z_{L(qN;1)\times S^1}[B]}{Z_{L(qN;1)\times S^1}[0]} &\sim \exp\left(-\min_{C\sim \gamma} \mathrm{Length}(C) \times (\text{effective monopole mass})\right),
    \label{eq:monopole_motivation}
\end{align}
where $\min_{C\sim \gamma} \mathrm{Length}(C)$ stands for the minimal length of cycles that belong to the homology class $[\gamma]$. 
If these two expectations are justified, the weak versions of our proposals are physically well motivated.

In addition, if the above expectations hold, the symmetry defects in our setups could be interpreted as center-vortex and monopole ``creation operators.''
The strong versions of our proposals indicate that all the ground states have to be invariant under such operations.
Hence, these are also reasonable definitions for condensations.

In the rest of this section, we illustrate that the above expectations on the twisted partition functions are valid.
To this end, we compute these symmetry-twisted partition functions for concrete gauge-theory models where the gauge symmetry is Higgsed, and we shall see that the leading nontrivial behaviors are indeed controlled by the center-vortex tension and the monopole mass, respectively. 

\subsubsection{Charge-\texorpdfstring{$N$}{N} Abelian-Higgs model}

Let us start with the easiest example, the charge-$N$ Abelian-Higgs model, which contain the $U(1)$ gauge field $a$ coupled with the charge-$N$ scalar field $\phi$: 
\begin{equation}
    S[a, \phi, B]=\int\left[\frac{1}{2g^2}|\diff a+\frac{2\pi}{N}B|^2+|(\diff+\im N a)\phi |^2 +V(|\phi|^2)\right]. 
\end{equation}
We take the wine-bottle potential, $V(|\phi|^2)=\lambda (|\phi|^2-v^2)^2$, with $v>0$, and the classical moduli is parametrized as $\phi(x)=v\, \rme^{\im \varphi(x)}$ with the $2\pi$-periodic scalar field $\varphi(x)$. 
Minimization of the scalar kinetic term is achieved by $Na+\diff \varphi=0$, and thus the Wilson loop obeys the perimeter law. This system is gapped and the $\mathbb{Z}_N^{[1]}$ symmetry is spontaneously broken. 

\subsubsection*{Torus twisted partition function}

This system has a vortex excitation. The vortex core is the codimension-$2$ surface specified by $\phi=0$. When sufficiently separated from the core, we can use the classical moduli condition, $\phi=v\, \rme^{\im \varphi}$ and $Na+\diff \varphi=0$, and we require that $\varphi$ has a winding number around the vortex core: $\frac{1}{2\pi}\int_C \diff \varphi=\pm 1$ for a sufficiently large loop $C$ that links with the vortex core with the winding number $1$. If we take a two-dimensional disc $D$ with $\partial D=C$, we find $\frac{1}{2\pi}\int_{D}\diff a=-\frac{1}{2\pi N}\int_C \diff \varphi=\mp \frac{1}{N}$, and thus the vortex core carries the fractional magnetic flux, which is the characterization of the center vortex.

Let us consider the torus partition function $Z_{T^4}[B]$ with the 't~Hooft flux $\int_{34}\frac{2\pi}{N}B=\frac{2\pi}{N}$, which requires the fractional magnetic flux along the $3$-$4$ direction. 
As the magnetic flux along the $3$-$4$ direction has to be $\frac{2\pi}{N}$ mod $2\pi$, this enforces the imbalance of the center vortex and anti-vortex excitations, and in particular, at least one dynamical center vortex is required to be present. This is identical with the domain-wall excitation for the Ising model with the symmetry-twisted boundary condition, and one can see more explicit computation in  Ref.~\cite[Appendix~B]{Tanizaki:2022ngt}.
The torus partition function then behaves as
\begin{align}
    \frac{Z_{T^4}[B]}{Z_{T^4}[0]}\sim \exp(-T_{\mathrm{vortex}} L_1L_2), 
\end{align}
where $T_{\mathrm{vortex}}$ is the vortex tension, and this is nothing but \eqref{eq:vortex_motivation}. 
Thus, our Defs.~\ref{def:center-vortex-condensation} or~\ref{def:center-vortex-condensation-strong}  suggest that the effective vortex tension becomes $0$, and thus the vacuum would be filled with the center-vortex configurations as expected from the analogue of Peiers's argument.

\subsubsection*{Lens-space twisted partition function}

Let us now consider the partition function on the lens space $L(q N;1) \times S^1$ with the background field $B$ specified by $\int_{\gamma\times S^1}\frac{2\pi}{N}B=\frac{2\pi}{N}$. This satisfies $\beta_{\mathbb{Z}}B=q\delta(\gamma)$ as discussed in \eqref{eq:derivaton_BocksteinHom} of Appendix~\ref{app:Lens_space}. 
We shall show that, for all $q\ge 1$,
\begin{align}
    \frac{Z_{L(qN;1)\times S^1}[B]}{Z_{L(qN;1)\times S^1}[0]}=0. 
\end{align}
That is, the lens-space twisted partition function vanishes exactly at finite volumes. 
This vanishing is a result of the absence of dynamical monopoles.
To illustrate this, let us provide two derivations: one concrete and one abstract. 

We start with the concrete derivation. For this purpose, it is convenient to consider the lattice regularization of the $U(1)$ gauge theory using the modified Villain formulation, which introduces 
\begin{itemize}
    \item $a_\ell$: $\mathbb{R}$-valued link variable, 
    \item $n_p$: $\mathbb{Z}$-valued plaquette variable, and 
    \item $\tilde{a}_{c}$: $\mathbb{R}$-valued cube (or dual link) variable. 
\end{itemize}
The $\mathbb{Z}_N$-valued $2$-form gauge field $B$ is realized as the $\mathbb{Z}$-valued plaquette variable $\hat{B}_p$ that satisfies $(\diff \hat{B})_c\in N\mathbb{Z}$. 
For the derivation, the scalar sector is not important, so let us here only write the lattice Maxwell action in this formulation,
\begin{equation}
    S^{(\mathrm{lat})}_{\mathrm{Max}}[a,n,\tilde{a},B]=\sum_p \frac{1}{2g^2}\Big|(\diff a)_p+2\pi \left(n_p+\frac{1}{N}\hat{B}_p\right)\Big|^2 + \im \sum_c \tilde{a}_c  \left((\diff n)_c+\frac{1}{N}(\diff \hat{B})_c\right). 
\end{equation}
We note that the ambiguity for taking the lift to the $\mathbb{Z}$-valued cochain $\hat{B}_p$ can be absorbed by the redefinition of the $\mathbb{Z}$-valued plaquette variable $n$, and thus the result does not depend on the choice of the lift. 
Since $\frac{1}{N}\diff \hat{B}=\beta_{\mathbb{Z}}B=q\delta(\gamma)$, the coupling $\tilde{a}\cup \beta(B)$ produces the 't~Hooft-like operator $\im q\int_{\gamma}\tilde{a}$,\footnote{Strictly speaking, however, we need to distinguish our computation from the computation of the one-point function of the 't~Hooft loop operator along $\gamma$ on $L(qN;1)\times S^1$. Importantly, we only introduce the symmetry operator for the $2$-cycle of $L(qN;1)$ in $H_2(L(qN;1);\mathbb{Z}_N)$,  and thus the dependence on $\gamma$ is topological and no renormalization issue appears, while the 't~Hooft loop operator geometrically depends on $\gamma$ and it is subject to the ultraviolet renormalization. } 
and the equation of motion of $\tilde{a}$ gives the delta-functional constraint of $\diff n+q\delta(\gamma)=0$. However, $q\gamma\subset L(qN;1)$ is not a boundary of any $2$-surfaces, so such configuration for $n$ does not exist and the twisted partition function vanishes identically.

Let us give a more formal derivation. Our model has the magnetic $U(1)^{[1]}$ symmetry in addition to the electric $\mathbb{Z}_N^{[1]}$ symmetry, and let us introduce the $U(1)$ $2$-form background gauge field $B_M$, which adds the topological coupling, 
\begin{equation}
    \im \int B_M \wedge \left(\frac{1}{2\pi}\diff a + n+\frac{1}{N}B\right). 
\end{equation}
This coupling is manifestly invariant under the electric $1$-form gauge transformation, $B\mapsto B+\diff \lambda_E^{(1)}$. 
Under the magnetic $1$-form gauge transformation, $B_M\mapsto B_M+\diff \lambda_M^{(1)}$, however, the partition function shows the 't~Hooft anomaly, 
\begin{align}
    Z[B_M+\diff \lambda_M^{(1)},B_E]=Z[B_M,B_E]\exp\left(\im \int_{\mathrm{PD}[\beta_{\mathbb{Z}}B]} \lambda^{(1)}_M \right),  
\end{align}
where $\mathrm{PD}[\beta_{\mathbb{Z}}B]$ is the Poincar\'e-dual $1$-cycle of $\beta_{\mathbb{Z}}B$. 
Especially by setting $B_M=0$ and $\diff \lambda_M^{(1)}=0$ for the lens space, $Z[B]=Z[0,B]$ satisfies
\begin{align}
    Z_{L(qN;1)\times S^1}[B]=Z_{L(qN;1)\times S^1}[B]\exp\left(\frac{2\pi\im}{N}\int_{\gamma} \lambda^{(1)}_M \right),  
\end{align}
and we reproduce $Z_{L(qN;1)\times S^1}[B]=0$. 

\subsubsection{\texorpdfstring{$SU(N)$}{SU(N)} gauge-adjoint Higgs model}

For the charge-$N$ Abelian-Higgs model, we have seen that the lens-space twisted partition function vanishes identically at any finite volumes, and the formal argument clarifies that it comes out of the presence of the magnetic $1$-form symmetry. 
Let us then consider the UV model that does not have the magnetic $1$-form symmetry and discuss the behavior of the lens-space twisted partition function. 

Here, we focus on the lens-space twisted partition function for the $SU(N)$ gauge theory with one adjoint scalar field $\Phi$, and we consider the case its vacuum expectation value is given by $\langle \Phi\rangle=\mathrm{diag}(v_1,\ldots, v_N)$ with $v_1>v_2>\cdots>v_N$, which causes the adjoint Higgsing,  $SU(N)\xrightarrow{\mathrm{Higgs}}U(1)^{N-1}$. 
The low-energy field theory belongs to the Coulomb phase, where the monopoles are well-defined thanks to the Abelianization at low energies.\footnote{Here, we discuss the case where the Abelianization of the gauge group is caused by the adjoint Higgs field, but this is not the essential assumption. The following discussion holds so far as the low-energy gauge group becomes $U(1)^{N-1}$ by other mechanisms, such as by choosing the maximal Abelian gauge.}  
One may add another adjoint Higgs field that further causes $U(1)^{N-1}\xrightarrow{\mathrm{Higgs}}\mathbb{Z}_N$, and then the low-energy field theory becomes identical with the one for the charge-$N$ Abelian-Higgs model studied above. 
Our purpose here is just to identify the leading non-vanishing contribution for $Z_{L(qN;1)\times S^1}[B]$, so let us keep our setup as simple as possible. 

The low-energy massless degrees of freedom are given by the $U(1)^{N-1}$ gauge field. 
Let us again use the modified Villain expression, and we use the $\mathbb{R}^{N-1}$-valued $1$-form fields, 
$\vec{a}=\sum_i a_i \vec{\alpha}_i$ and $\vec{\tilde{a}}=\sum_i \tilde{a}_i \vec{\mu}_i$, and the $\Lambda_{\mathrm{root}}$-valued $2$-form field, $\vec{n}=\sum_i n_i \vec{\alpha}_i$.\footnote{For completeness, we here summarize the convention of the $SU(N)$ root and weight vectors using the canonical orthonormal basis of $\mathbb{R}^N$: $\{\vec{e}_n\}_{n=1,\ldots, N}$.
The positive simple roots are given by $\vec{\alpha}_n = \vec{e}_n-\vec{e}_{n+1}~(n=1,\cdots,N-1)$, and they span the root lattice $\Lambda_{\mathrm{roots}} := \{ \vec{\alpha} \in \mathbb{Z}^N~|~\sum_{i=1}^N \alpha_i =0 \}$.
The weight lattice $\Lambda_{\mathrm{weights}}$ is spanned by the fundamental weights $\vec{\mu}_n$ ($n=1,\ldots, N-1$), which are given by $\vec{\mu}_n=\vec{e}_1+\cdots+\vec{e}_n-\frac{n}{N}\sum_{k=1}^{N}\vec{e}_k$.
} 
Then, the coupling to the background $\mathbb{Z}_N$ $2$-form gauge field $B$ is given by 
\begin{align}
    S_{\mathrm{eff}}=\frac{1}{2g^2}\int |\diff \vec{a}+2\pi \vec{n}+2\pi \vec{\mu}_1 \hat{B}|^2+\frac{\im}{2\pi}\int \diff \vec{\tilde{a}}\wedge (\diff \vec{a}+2\pi\vec{n}+2\pi\vec{\mu}_1 \hat{B}),
\end{align}
where we take a lift of $B$ to a $\mathbb{Z}$-valued $2$-cochain $\hat{B}$. 
To be more precise, this lift is associated with the short exact sequence,
\begin{equation}
    0 \to \Lambda_{\mathrm{roots}} \to \Lambda_{\mathrm{weights}} \to \mathbb{Z}_N \to 0,
\end{equation}
and the $\mathbb{Z}_N^{[1]}$-symmetry background $B \in H^2(M_4; \mathbb{Z}_N)$ can be represented by a 2-cochain $\vec{\mu}_1\hat{B} \in C^2(M_4; \Lambda_{\mathrm{weights}})$ in terms of the $U(1)^{N-1}$ magnetic flux. 
There exists ambiguity for the choice of the lift, but it can be compensated by the redefinition of $\vec{n}\in C^2(M_4;\Lambda_{\mathrm{root}})$. 
% Taking its Abelian duality, the theory can be written in terms of the dual $U(1)^{N-1}$ gauge field, $\vec{\tilde{a}}=\sum_i \tilde{a}_i \mu_i$, as 
% \begin{align}
%     S_{\mathrm{eff}}=\frac{1}{2\tilde{g}^2}\int |\diff \tilde{a}|^2+ \frac{\im }{2\pi}\int \diff \vec{\tilde{a}} \cdot 2\pi (\mu_1+\alpha) B. 
% \end{align}
On $L(q N;1)\times S^1$ with $\beta_{\mathbb{Z}}B=\frac{1}{N}\diff \hat{B}=q\delta(\gamma)$, the topological coupling can be calculated as follows:
\begin{align}
    \frac{\im }{2\pi}\int \diff \vec{\tilde{a}} \cdot 2\pi \vec{\mu}_1 \hat{B}
    &= \im \int \vec{\tilde{a}}\cdot N\vec{\mu}_1\wedge \frac{1}{N} \diff \hat{B} 
    = \im q \int_\gamma \vec{\tilde{a}}\cdot N\vec{\mu}_1. 
\end{align}
As we have seen in the charge-$N$ Abelian-Higgs model, the partition function vanishes within this effective theory as $q\gamma$ is not a boundary of any $2$-surfaces. 

In the full gauge-adjoint Higgs theory, however, there exist dynamical 't~Hooft-Polyakov monopoles with the magnetic charges $\alpha_i$ ($i=1,\ldots, N-1$)~\cite{tHooft:1974kcl, Polyakov:1974ek}. 
When we consider the excitation of the $\vec{\alpha}_i$-monopole along the worldline $C_i$, it contributes to the action as $M_{\alpha_i}\mathrm{Length}(C_i)+\im \int_{C_i}\vec{\tilde{a}}\cdot \vec{\alpha}_i$, where $M_{\alpha_i}$ is the $\vec{\alpha}_i$-monopole mass. 
The equation of motion in terms of $\vec{\tilde{a}}_i$ including these monopole worldlines gives the delta-function constraint of 
\begin{align}
    \diff \vec{n}+q N\vec{\mu}_1 \delta(\gamma)+\sum_i \vec{\alpha}_i \delta(C_i)=0, 
\end{align}
which can be solved for $\vec{n}$ iff $[q N\vec{\mu}_1 \gamma +\sum_i \vec{\alpha}_i C_i]=0\in H_1(L(Nq;1)\times S^1;\Lambda_{\mathrm{root}})$. 
An obvious sufficient condition is to set $C_i=-(N-i)q\gamma$ since $N\vec{\mu_1}=\sum_i (N-i)\vec{\alpha}_i$ so that the cancellation occurs on the nose, but one can freely add/subtract $Nq \gamma$ for each $C_i$ since $qN \vec{\alpha_i} \gamma$ is a boundary of $2$-surface for the corresponding homology. 
We can also change the shape of each $C_i$ continuously as it does not change the homology element.
Thus, the dominant contribution comes from the monopole sector whose magnetic charge is $Nq\vec{\mu}_1$ up to $Nq\Lambda_{\mathrm{root}}$ with the shortest worldline $C$, which is homologically same with $\gamma$.  

Let us denote $M_{\mathrm{mon}.}(\vec{\alpha})$ for the lowest mass gap for the monopole charge $\vec{\alpha}\in \Lambda_{\mathrm{root}}$, 
and then the leading nontrivial behavior of the lens-space twisted partition function can be obtained as
\begin{align}
    \frac{Z_{L(qN;1)\times S^1}[B]}{Z_{L(qN;1)\times S^1}[0]}\sim \exp\left(-\min_{\vec{\alpha}'\in \Lambda_{\mathrm{root}}}M_{\mathrm{mon}.}(qN\vec{\mu}_1+qN\vec{\alpha}')\times \min_{C\sim \gamma} \mathrm{Length}(C)\right). 
\end{align}
This is nothing but \eqref{eq:monopole_motivation}. 
As a result, our criterion of the charge-$q$ monopole condensation, Def.~\ref{def:lens-space-criterion}, requires that the monopole mass becomes zero for the magnetic charge $qN\vec{\mu}_1$ modulo $qN\Lambda_{\mathrm{root}}$. 
The above observation is summarized as follows:
\begin{itemize}
    \item We consider an abelianization scenario, where the $SU(N)$ gauge group is reduced into $U(1)^{N-1}$.
    \item The $\mathbb{Z}_N$ one-form symmetry background $B$ is interpreted as a background magnetic flux $\vec{\mu}_1\hat{B}$ in the weight lattice $\Lambda_{\mathrm{weights}}$.
    \item When the Bockstein homomorphism $\beta(B)$ is non-trivial, the Bianchi identity for the magnetic flux must be violated: $\diff \hat{B} \neq 0$. With such a background, the existence of magnetic monopoles is forced when the abelianization occurs.
\end{itemize}
Hence, the non-suppression of the partition function under the twisted boundary condition, which requires the existence of a magnetic flux source, would be suitable as a criterion for monopole condensation.
When $H^3(M_4; \mathbb{Z})$ has no torsion, the Bockstein homomorphism is trivial, and we cannot force the presence of monopoles, which is why we choose $M_4 = L(qN;1) \times S^1$ in Def.~\ref{def:lens-space-criterion}, where the Bockstein homomorphism becomes nontrivial.
% See Appendix \ref{app:Lens_space} for details of the lens space.

By the above discussion, the requirement of Def.~\ref{def:lens-space-criterion} is that the lattice of condensing magnetic charges should have an overlap with $Nq\vec{\mu}_1+Nq \Lambda_{\mathrm{root}}$. 
Thus, even if the fundamental condensing monopoles have smaller magnetic charges, say $\vec{\alpha}_i$'s, it satisfies the charge-$q$ monopole condensation criterion. 
Moreover, the charge-$q$ monopole condensation criterion can be satisfied even if the fundamental condensing magnetic charge is $Nq(\vec{\mu}_1 + k_i \vec{\alpha}_i)$ with some $k_i\gg 1$. 
This subtlety is the motivation behind our terminology of charge-$[q]_{Nq\mathbb{Z}}$ monopole condensation. 
We will revisit the nontrivial condensation pattern in Section \ref{sec:highercharge}.

\subsection{Relation between center-vortex and monopole criteria for gapped phases}
\label{sec:main_result}

In this section, we reveal the relationship between the center-vortex and monopole condensation in the weak versions (Claim \ref{claim:weak-vortex-monopole-criterion}) and also in the strong versions (Claim \ref{claim:strong-vortex-monopole-criterion}).

\subsubsection{Equivalence for the weak center-vortex/monopole condensations}

In this section, we show the equivalence of the weak versions of the center-vortex condensation criterion (Def.~\ref{def:center-vortex-condensation}) and the monopole condensation criterion (Def.~\ref{def:lens-space-criterion}) in gapped phases.
We note that the remote detectability for the low-energy TQFT does not play any essential role for the proof of this statement.

\begin{claim}
\label{claim:weak-vortex-monopole-criterion}
For the $(3+1)$d local, unitary, $\mathbb{Z}_N^{[1]}$-symmetric TQFTs with remote detectability, 
\begin{align}
    &~~~~~ Z^{(\mathrm{TQFT})}_{T^4}[B=\delta(x_3,x_4)\diff x_3\wedge \diff x_4]\not=0 \notag\\
    &\Leftrightarrow \exists q>0 ~~\text{s.t.}~~ Z^{(\mathrm{TQFT})}_{L(qN;1)\times S^1}[B=A_{\mathbb{Z}}\wedge \delta(x_4)\diff x_4]\not = 0. 
\end{align}
This establishes the equivalence between the center-vortex/monopole condensation criteria under the presence of the mass gap: 
\begin{align}
   &~~~~~ \text{weak center-vortex condensation (Def.~\ref{def:center-vortex-condensation})} \notag \\
& \Leftrightarrow  
    \text{weak charge-$[q]_{Nq\mathbb{Z}}$ monopole condensation for some $q>0$ (Def.~\ref{def:lens-space-criterion})}  .
\end{align}
\end{claim}

\begin{proof}[Derivation of Claim 1]

As a first step, let us show the center-vortex condensation criterion on $T^4$ (in the weak version) is equivalent to the nonvanishing of the solid torus state in the low-energy effective theory:
\begin{align}
    0 \neq \ket{Z_{D^2\times T^2}[B']} \in \mathcal{H}_0({T^3, ~\text{twisted}}),  \label{eq:nonvanishing_solid torus}
\end{align}
where $\mathcal{H}_0({T^3, ~\text{twisted}})$ denotes the twisted Hilbert space with the defect intersecting the first $S^1$ of $T^3$, and $B'$ represents the minimal $\mathbb{Z}_N^{[1]}$ defect on the open disk $D^2$.
Indeed, the locality principle implies that the twisted partition function on $T^4$ in Def.~\ref{def:center-vortex-condensation} can be written as
\begin{align}
   Z_{T^4}[B=\delta(T_{(12)}^2)] = \braket{Z_{(T^2 \setminus D^2)\times T^2 }[B'']~|~Z_{ D^2\times T^2}[B']},
\end{align}
where $B'$ and $B''$ represent the defect $B$ restricted on $D^2$ and $(T^2 \setminus D^2)$, respectively.
This relation yields $\ket{Z_{D^2\times T^2}[B']}  \neq 0$ from Def.~\ref{def:center-vortex-condensation}.
Conversely, the nonvanishing of the solid torus state $\ket{Z_{D^2\times T^2}[B']}  \neq 0$ implies that the twisted Hilbert space is nontrivial, $\dim \mathcal{H}_0({T^3, ~\text{twisted}}) > 0$, and the twisted partition function on $T^4$, which is $\dim \mathcal{H}_{T^3, ~\text{twisted}}$ in the infrared, does not vanish.
This establishes the equivalence between Def.~\ref{def:center-vortex-condensation} and (\ref{eq:nonvanishing_solid torus}).

Next, let us give a relationship between the solid torus state and the lens space.
Topologically, the lens space can be constructed by gluing two solid tori $D^2 \times S^1$ with a Dehn twist: We take the coordinate $(x_1,x_2)\in T^2=\partial(D^2\times S^1)$ with the periodicity $x_{1,2}\sim x_{1,2}+1$, and consider the diffeomorphism $f_{\mathrm{Dehn}}:T^2\to T^2, (x_1,x_2)\mapsto (x_1, x_2+Nx_1)$. 
Then, we obtain $L(qN;1)=(D^2\times S^1)\cup_{f_{\mathrm{Dehn}}^q}(D^2\times S^1)$, where $f^{q}_{\mathrm{Dehn}}=f_{\mathrm{Dehn}}\circ \cdots \circ f_{\mathrm{Dehn}}$.  
In this construction, the non-trivial $\mathbb{Z}_N^{[1]}$ defect $B$ with $\beta B \neq 0$ corresponds to introducing the $\mathbb{Z}_N^{[1]}$ defect on the open disk $D^2$ for both solid tori, which is nothing but $B'$. 
Since the Dehn twist $(f_{\mathrm{Dehn}},\mathrm{id}_{S^1}):T^2\times S^1\to T^2\times S^1$ gives a diffeomorphism on $T^3$ and maintains the $\mathbb{Z}_N^{[1]}$-twisted boundary condition by $B'$ unchanged, it induces the unitary operator 
\begin{align}
    \hat{U}_{\mathrm{Dehn}}: \mathcal{H}_0(T^3, \mathrm{twisted})\to \mathcal{H}_0(T^3,\mathrm{twisted})
\end{align}
for the Hilbert space of the local, unitary TQFTs. 
Equivalently, one may view the Dehn twist as a ``time evolution'' along the mapping cylinder, so the Dehn twist trivially induces a unitary operator on the TQFT Hilbert space.
Applying the locality principle for this Dehn-twist construction of the lens space, 
we obtain 
\begin{align}
   Z_{L(qN;1)\times S^1}[B] = \bra{Z_{D^2\times T^2}[B']} (\hat{U}_{\mathrm{Dehn}})^{q}\ket{Z_{D^2\times T^2}[B']}
\end{align}
for $q>0$. 
If one sets $q=0$, one may obtain the  twisted partition function on $S^2\times T^2$.

Now, let us show the equivalence claim. 
First, assuming the center-vortex condensation criterion, we derive the monopole condensation criterion. 
The twisted Hilbert space $\mathcal{H}_0({T^3, ~\text{twisted}})$ has finite dimension, and let us write it as $D=\dim \mathcal{H}_0({T^3, ~\text{twisted}}) < \infty$.
Using the eigenvalues $(\rme^{\im \theta_1}, \rme^{\im \theta_2}, \cdots, \rme^{\im \theta_D})$ of $\hat{U}_{\mathrm{Dehn}} $, 
the lens-space twisted partition function can be written in the following form,
\begin{align}
   Z_{L(qN;1)\times S^1}[B] = \sum_{j=1}^D c_j \rme^{\im q \theta_j},
\end{align}
with some nonnegative constants $\{ c_j\}_{j=1,\cdots,D}$. 
If this was identically equal to $0$ for all $q>0$, it implies $c_j=0$ and contradicts $\ket{Z_{D^2\times T^2}[B']}\not=0$. 
Therefore, at least for some $q>0$, the twisted partition function of the low-energy effective theory satisfies $Z_{M_4 =L(qN;1)\times S^1}[B] \neq 0$, which is nothing but the monopole condensation criterion.

The converse is immediate.
If $Z_{M_4 =L(qN;1)\times S^1}[B] \neq 0$ for some $q$, the solid torus state $\ket{Z_{D^2\times T^2}[B']}$ cannot vanish, which is identical to the weak-version of the center-vortex condensation criterion on $T^4$.
\end{proof}

One might suspect that a counterexample can be constructed if the TQFT with an infinite-dimensional Hilbert space ($D=\infty$) is allowed.
However, as we discuss in the next section~\ref{sec:center-vortex-model}, such a vacuum can be understood as a strong-coupling limit of the gapless Coulomb phase. 
Without finite dimensionality, the distinction between gapped and gapless phases is no longer clear, 
and we strictly reserve the term ``gapped phase'' for cases where the low-energy theory is governed by a well-behaved finite-dimensional TQFT.

\subsubsection{Relation between the strong center-vortex/monopole condensations}

Here, we show that the strong center-vortex condensation implies the strong monopole condensation. Furthermore, the strong monopole condensation turns out to be equivalent with the unbroken $\mathbb{Z}_N^{[1]}$ symmetry. 
For this proof, the structural theorem for the $(3+1)$d low-energy TQFTs with the remote detectability plays the essential role. 

\begin{claim}
\label{claim:strong-vortex-monopole-criterion}
For the $(3+1)$d gapped QFTs with the $\mathbb{Z}_N^{[1]}$ symmetry, we assume that the low-energy TQFT satisfies the remote detectability. Then, 
\begin{align}
   &~~~~~ \text{strong center-vortex condensation (Def.~\ref{def:center-vortex-condensation-strong})} \notag \\
& \Rightarrow  \text{Unbroken $\mathbb{Z}_N^{[1]}$ symmetry}\notag\\
& \Leftrightarrow    \text{strong charge-$[q]_{Nq\mathbb{Z}}$ monopole condensation for some $q>0$ (Def.~\ref{def:lens-space-criterion-strong})}  .
\end{align}
\end{claim}

\begin{proof}[Derivation of Claim 2] 
Since we have already seen that both the strong center-vortex and monopole condensations imply the unbroken $\mathbb{Z}_N^{[1]}$ symmetry, 
the remaining problem is to prove the strong monopole condensation from the unbroken $\mathbb{Z}_{N}^{[1]}$ symmetry. 

The $(3+1)$d local, unitary TQFTs with remote detectability are described by finite $2$-group gauge theories~\cite{Lan:2018vjb, Lan:2018bui,Johnson-Freyd:2020usu}. 
Thanks to the duality~\cite{Kapustin:2013uxa}, the $(3+1)$d $2$-group gauge theory has a canonical description with the $1$-form discrete $G$ gauge field $a$ (and also with a $\mathbb{Z}_2$ $2$-form gauge field if some excitations are fermionic)~\cite{Thorngren:2020aph}. 
Assuming the unbroken $\mathbb{Z}_N^{[1]}$ symmetry, all the line operators in the low-energy TQFT commute with the $\mathbb{Z}_N^{[1]}$ symmetry operator. 
Then, the possible coupling for the $\mathbb{Z}_N^{[1]}$ symmetry to the low-energy higher-gauge theory is described by the symmetry fractionalization~\cite{Chen:2014wse, Barkeshli:2014cna, Hsin:2019fhf, Delmastro:2022pfo, Hsin:2024aqb, Brennan:2025acl}, which is described by the magnetic-type coupling with the background gauge field $B$~\cite{Thorngren:2020aph}: One of the possibility is given by the form, 
\begin{align}
    \frac{2\pi\im}{N}\int_{M_4} B\cup \eta(a),
    \label{eq:magnetic_coupling_eta}
\end{align}
where $\eta\in H^2(\sfB G, \mathbb{Z}_N)$ and $\eta(a)\in H^2(M_4,\mathbb{Z}_N)$ is its pullback by the $G$ gauge field $a:M_4\to \sfB G$. 
The other possibility is given by the form, 
\begin{align}
    2\pi \im \int \langle\theta(B), a\rangle, 
    \label{eq:magnetic_coupling_Bock}
\end{align}
where $\theta\in H^3(\sfB^2\mathbb{Z}_N, \hat{G})\simeq \mathrm{Ext}_{\mathbb{Z}}(H_2(\sfB^2\mathbb{Z}_N), \hat{G})$ with $\hat{G}=\mathrm{Hom}(G,U(1))\simeq \prod_i \mathbb{Z}_{m_i}$ and $\langle , \rangle$ is the canonical pairing $\hat{G}\times G\to U(1)$.

Now, let us consider the lens-space twisted partition function on $L(qN;1)\times S^1$, and we have chosen $B=\delta(\Sigma)$ for $\Sigma\in H_2(L(qN;1); \mathbb{Z}_N)\subset H_2(L(qN;1)\times S^1; \mathbb{Z}_N)$ such that $\beta_{\mathbb{Z}}(B)=q\delta(\gamma)$ with $\gamma\in H_1(L(qN;1))\simeq \mathbb{Z}_{qN}$ following the setup of Def.~\ref{def:lens-space-criterion-strong}. 
We will show that both magnetic couplings are trivial for $q=\mathrm{gcd}(\lcm_{g\in G}(\ord(g)),N)$, which proves $\hat{U}(\Sigma)=1$ on $\mathcal{H}_0(L(qN;1))$ and then the strong monopole condensation follows from \eqref{eq:lens_partition_function_trace}.

Let us first work on the magnetic coupling $2\pi \im\int \langle \theta(B), a\rangle$. 
We should note that $H^3(\sfB^2\mathbb{Z}_N,\mathbb{Z})\simeq \mathrm{Ext}(H_2(B^2\mathbb{Z}_N),\mathbb{Z})\simeq \mathbb{Z}_N$, so the Bockstein map $\beta_{\mathbb{Z}}:H^2(\sfB^2\mathbb{Z}_N,\mathbb{Z}_N) \to H^3(\sfB^2\mathbb{Z}_N,\mathbb{Z})$ gives an isomorphism. 
Since $H^3(\sfB^2\mathbb{Z}_N, \hat{G})\simeq \mathrm{Ext}(\mathbb{Z}_N,\hat{G})\simeq \hat{G}/N\hat{G}$, 
we can obtain $\theta\in H^3(\sfB^2\mathbb{Z}_N,\hat{G})$ by sending a generator of $H^3(\sfB^2\mathbb{Z}_N,\mathbb{Z})$ with choosing a homomorphism $\phi_{\theta}:\mathbb{Z}\to \hat{G}$ (non-canonically). %such that $\phi_{\theta}(N)=0\in \hat{G}$.
Thus, we obtain $\theta(B)\in H^3(M_4,\hat{G})$ as 
\begin{align}
    \theta(B)=\phi_{\theta}(\beta_{\mathbb{Z}}B). 
\end{align}
Substituting this expression for our lens-space setup, we find 
\begin{align}
    2\pi \im \int_{L(qN;1)\times S^1}\langle \theta(B),a\rangle 
    &= 2\pi \im q\underbrace{\int_{\gamma}\langle\phi_{\theta}(1),a\rangle}_{\in \frac{1}{\gcd(N,\mathrm{lcm}(m_i))}\mathbb{Z}},
\end{align}
where we use $\beta_{\mathbb{Z}}B=q\delta(\gamma)\in\mathbb{Z}_{qN}$. 
The quantization of the integral comes from $\hat{G}\simeq \prod_i \mathbb{Z}_{m_i}$ and $B:M_4\to \sfB^2\mathbb{Z}_N$. 
Thus, if $q$ is a multiple of $\gcd(N,\mathrm{lcm}(m_i))$, this magnetic coupling becomes trivial, and $q=\gcd(\lcm_{g\in G}(\ord(g)),N)$ satisfies it. 

Next, let us consider the other magnetic coupling,
\begin{align}
    \frac{2\pi\im}{N}\int_{L(qN;1)\times S^1} B\cup \eta(a) = \frac{2\pi \im}{N}\int_{\Sigma}\eta(a). 
\end{align}
Since $\Sigma\subset L(qN;1)\times \{\mathrm{pt}\}$, this integral depends only on the restriction of $a$ to the lens space $L(qN;1)$ of the specific time slice, so let us forget about the $S^1$ factor below. 
Then, the flat $G$ gauge field $a$ on $L(qN;1)$ is totally determined by the holonomy, $g=\mathcal{P}\exp\im\int_\gamma a\in G$. 
Considering the cyclic group generated by $g$, $\mathbb{Z}_{\ord(g)}\simeq \langle g\rangle\subset G$, 
$a$ can be regarded as the $\mathbb{Z}_{\ord(g)}$ gauge field after a suitable $G$ gauge transformation, and we can think of $\eta\in H^2(\sfB G, \mathbb{Z}_N)$ as an element of $ H^2(\sfB \mathbb{Z}_{\ord(g)},\mathbb{Z}_N)$ to compute $\eta(a)\in H^2(M_3,\mathbb{Z}_N)$ as long as $a:M_3\to \sfB\mathbb{Z}_{\ord(g)}\subset \sfB G$.\footnote{Here, it is crucial that we have forgotten the $S^1$ factor as $\Sigma$ localizes on the specific time slice. 
The $4$-manifold $L(qN;1)\times S^1$ also has a $2$-surface $\Sigma'=\gamma\times S^1$, and if we evaluate $\eta(a)$ on $\Sigma'$, the logic here does not apply: The presence of two independent $1$-cycles can give the topological invariant by the cup product for $\eta(a)$. 
Indeed, we believe that the lens-space twisted partition function with $B=b_1 \delta(\Sigma)+b_2 \delta(\Sigma')$ can fully detect both of the magnetic couplings~\eqref{eq:magnetic_coupling_eta} and \eqref{eq:magnetic_coupling_Bock} by changing the values of $q$. 
}  In our context, $M_3=L(qN;1)\times \{\mathrm{pt}\}\subset L(qN;1)\times S^1$.

This $\eta \in H^2(\sfB \mathbb{Z}_{\ord(g)}, \mathbb{Z}_N)$ is proportional to the Bockstein element. 
To show this, we notice that the Bockstein map $\hat{\beta}_{\mathbb{Z}}: H^1(\sfB\mathbb{Z}_{\ord(g)}, \mathbb{Z}_{\ord(g)}) \rightarrow H^2(\sfB\mathbb{Z}_{\ord(g)}, \mathbb{Z})$ is an isomorphism.
As the mod-$N$ reduction $\rho_N$ is surjective, $\eta$ can be written as $\eta = \rho_N(\hat{\beta}_{\mathbb{Z}}(c\, \iota))$ for some integer $c \in \mathbb{Z}$, where $\iota \in H^1(\sfB\mathbb{Z}_{\ord(g)}, \mathbb{Z}_{\ord(g)})$ is the generator.
Whereas $a: M_3 \to \sfB\mathbb{Z}_{\ord(g)}$, it is the convention to use the same notation for $a = a^*\iota \in H^1(M_3, \mathbb{Z}_{\ord(g)})$. In this notation, we find that 
\begin{align}
    \eta(a) = a^* \eta = c\, \beta_{\mathbb{Z}}(a)~(\operatorname{mod}N),
\end{align}
with the Bockstein map $\beta_{\mathbb{Z}}: H^1(M_3, \mathbb{Z}_{\ord(g)}) \rightarrow H^2(M_3)$ associated with $\mathbb{Z} \rightarrow \mathbb{Z} \rightarrow \mathbb{Z}_{\ord(g)}$.

Meanwhile, this Bockstein $\beta_{\mathbb{Z}}$ is essentially the homomorphism from $\mathbb{Z}_{\ord(g)} $ to $\mathbb{Z}_{Nq}$, from which we can write
\begin{align}
   \beta_{\mathbb{Z}}(a) = \frac{N q \ell }{\gcd(\ord(g),Nq)} \in  \mathbb{Z}_{Nq} = H^2(M_3),
\end{align}
with some $\ell \in \mathbb{Z}$.
This indicates that, if $q$ is a multiple of $\gcd(\ord(g),N)$, $\beta_{\mathbb{Z}}(a)$ has a factor of $N$.
Therefore, we have
\begin{align}
    \eta(a) = 0 \in H^2(M_3,\mathbb{Z}_N) ~~~\text{if}~\gcd(\ord(g),N)|q.
\end{align}
This indeed establishes that this fractionalization becomes trivial when $q$ is a multiple of $\gcd(\ord(g),N)$ for all $g \in G$, and thus we can take $q=\gcd(\lcm_{g\in G}(\ord(g)),N)$.
\end{proof}

\subsection{Comment on the unique gapped vacuum}

The confined phases are often defined as the $\mathbb{Z}_N^{[1]}$ unbroken phases.
Typically, e.g., within the conventional Wilson-'t Hooft classification framework~\cite{tHooft:1979rtg, Nguyen:2023fun}, the confined phase is assumed to have the uniquely gapped ground state on any closed spatial $3$-manifolds, or, equivalently, the deconfined nontrivial lines are totally absent whether or not they are electrically charged. 
One can judge if this condition is satisfied by looking at the Gu-Wen ratios~\eqref{eq:GuWen_ratio}, and it is sufficient to check it for $Y_3=T^3$ for local, unitary TQFTs with the remote detectability: 
If the ground state degeneracy on $T^3$ is one, the underlying TQFT is necessarily invertible and reduces to a symmetry-protected topological (SPT) phase, see e.g. \cite{Lan:2018vjb}.

Then, the possible low-energy behavior of $Z_{M_4}[B]$ is completely determined up to the gravitational counterterms and described by the classical topological action for the background gauge field: For the $\mathbb{Z}_N^{[1]}$-symmetry background $B$, such a possible response is 
\begin{align}
    Z_{M_4}[B]=\exp\left[\frac{2\pi \im k}{N}\int_{M_4}\frac{1}{2}P_2(B)\right],  \label{eq:SPTphase}
\end{align}
with $k\in \mathbb{Z}_N$ and $P_2(B)=\hat{B}\cup \hat{B}+\hat{B}\cup_1\diff \hat{B}$ is the Pontryagin square.

The condition of the unique gapped vacuum is the strongest condition compared to other criteria proposed in this paper: When the nonzero mass gap is present, 
\begin{align}
    &\text{Unique gapped vacuum} \notag\\
    \Rightarrow &~
    \text{Def.}~\ref{def:center-vortex-condensation-strong}~~\text{(strong  center-vortex condensation)} \notag\\
    \Rightarrow
    & 
    \left\{\begin{array}{l}
        \text{Unbroken $\mathbb{Z}_N^{[1]}$ symmetry} ~~\Leftrightarrow \\
        \text{Def.}~\ref{def:lens-space-criterion-strong}~~\text{(strong charge-$[q]_{Nq\mathbb{Z}}$ monopole condensation for some $q>0$)}  
    \end{array}\right.\notag\\
    \Rightarrow& 
    \left\{\begin{array}{l}
        \text{Def.}~\ref{def:center-vortex-condensation}~~ \text{(weak center-vortex condensation)}~~\Leftrightarrow\\
         \text{Def.}~\ref{def:lens-space-criterion}~~\text{(weak charge-$[q]_{Nq\mathbb{Z}}$ monopole condensation for some $q>0$)}.
    \end{array}\right.  \label{eq:summary_relation}
\end{align}
Indeed, from the SPT response to the background field (\ref{eq:SPTphase}), it is immediate that both center-vortex condensation and monopole condensation are satisfied for the unique gapped phase as $|Z_{Y_3\times S^1}[B]|=Z_{Y_3\times S^1}[0] \neq 0$ for any choice of $Y_3$ and $B$.
Following our criterion, the unique gapped phase has both center-vortex and monopole condensation.

\section{Center-vortex model with monopoles}
\label{sec:center-vortex-model}

\subsection{Formal-sum representation of center-vortex model}
\label{sec:formal-sum_center-vortex}

Since the main argument is rather abstract, it would be useful to explicitly examine certain aspects within the framework of the center-vortex percolating model with monopole worldlines~\cite{DelDebbio:1997ke, Ambjorn:1999ym, deForcrand:2000pg, Engelhardt:1999wr, Oxman:2018dzp, Junior:2019fty, Junior:2022bol}. 
Ignoring many details (such as thickness of center vortices), the center-vortex model can be formally represented as follows: 
The dynamical degrees of freedom (under some lattice regularization) are,
\begin{itemize}
    \item 2-chain $\vec{\Sigma} \in C_2(M_4; \Lambda_{\mathrm{weights}})$: center-vortex worldsheets (localized magnetic flux),
    \item 1-cycle $\vec{C} \in C_1(M_4; \Lambda_{\mathrm{roots}})$: monopole worldlines, satisfying $\vec{C} = \partial \vec{\Sigma}$,
\end{itemize}
and the partition function is given by
\begin{align}
    Z_{\text{model}}= \sum_{\substack{\vec{\Sigma} \in C_2(M_4; \Lambda_{\mathrm{weights}})  \\ \vec{C} \in C_1(M_4; \Lambda_{\mathrm{roots}}) \\
    \text{satisfying~} \vec{C} = \partial \vec{\Sigma}}} \rme^{-S_{\mathrm{vortex}}[\vec{\Sigma}] - S_{\mathrm{monopole}}[\vec{C}]} \left( \sum_{b \in H^2 (M_4;\mathbb{Z}_N)} \rme^{\frac{2\pi \im }{N}b(\vec{\Sigma})}\right),
    \label{eq:center-vortex-model-untwisted}
\end{align}
where $S_{\mathrm{vortex}}[\vec{\Sigma}]$ and $S_{\mathrm{monopole}}[\vec{C}]$ are some worldsheet or worldline weight functions (and we do not need to specify them for the discussion below).
This expression describes the ensemble comprising all possible configurations of center-vortex networks connected by monopole junctions.

The last factor with the dynamical $\mathbb{Z}_N$-valued $2$-form gauge field $b$ is introduced so that the global magnetic flux must be a root-vector one, i.e. the $N$-ality of the global magnetic flux is trivial since $\mathbb{Z}_N \simeq \Lambda_{\mathrm{weights}}/\Lambda_{\mathrm{roots}}$.
Without this constraint, i.e., if the global magnetic flux $\vec{\Sigma}$ is weight-vector-valued, the cocycle condition for the $SU(N)$ bundle is violated by $\mathbb{Z}_N$ center elements, which give the $PSU(N)$ gauge bundle instead of the $SU(N)$ bundle. 
One can pass from $PSU(N)$ gauge theory to $SU(N)$ gauge theory by gauging the magnetic 1-form symmetry $( \mathbb{Z}_N^{[1]} )_{\mathrm{mag}}$, which is nothing but the last factor of \eqref{eq:center-vortex-model-untwisted}. 
From this viewpoint, we can easily couple this model to the electric $\mathbb{Z}_N^{[1]} $ symmetry background $B \in H^2 (M_4; \mathbb{Z}_N)$:
\begin{align}
    Z_{\text{model}}[B]= \sum_{\substack{\vec{\Sigma} \in C_2(M_4; \Lambda_{\mathrm{weights}})  \\ \vec{C} \in C_1(M_4; \Lambda_{\mathrm{roots}}) \\
    \text{satisfying~} \vec{C} = \partial \vec{\Sigma}}} \rme^{-S_{\mathrm{vortex}}[\vec{\Sigma}] - S_{\mathrm{monopole}}[\vec{C}]} \left( \sum_{b \in H^2 (M_4;\mathbb{Z}_N)} \rme^{\frac{2\pi \im }{N}b(\vec{\Sigma}) - \frac{2 \pi \im}{N} (B \cup b)(M_4)}\right). 
    \label{eq:center-vortex-model}
\end{align}

To see the correspondence to the argument in the previous section, it is instructive to use the Poincar\'e dual $\Sigma_B \in H_2(M_4; \mathbb{Z}_N)$ and consider the lift to $\hat{\Sigma}_B \in C_2(M_4; \Lambda_{\mathrm{weights}})$: $\vec{\mu_1}\cdot \hat{\Sigma}_B=\frac{1}{N}\Sigma_{B}$ mod $1$.
If the Bockstein map $\beta_{\mathbb{Z}}B$ does not vanish, any lift $\hat{\Sigma}_B$ of $\Sigma_B$ is not a cycle: $\partial \hat{\Sigma}_B \neq 0$. Since the worldline $\partial \hat{\Sigma}_B $ has trivial $N$-ality, we can regard $\partial \hat{\Sigma}_B \in C_1(M_4; \Lambda_{\mathrm{roots}}) $.
Using this lift $\hat{\Sigma}_B$, we can rewrite $Z_{\text{model}}[B]$ by shifting $(\vec{\Sigma},\vec{C}) \mapsto (\vec{\Sigma}+\hat{\Sigma}_B,\vec{C}+\partial \hat{\Sigma}_B) $ and we obtain
\begin{align}
    Z_{\text{model}}[B] &= \sum_{\substack{\vec{\Sigma} \in C_2(M_4; \Lambda_{\mathrm{weights}})  \\ \vec{C} \in C_1(M_4; \Lambda_{\mathrm{roots}}) \\
    \text{satisfying~} \vec{C} = \partial \vec{\Sigma}}} \rme^{-S_{\mathrm{vortex}}[\vec{\Sigma}] - S_{\mathrm{monopole}}[\vec{C}]} \left( \sum_{b \in H^2 (M_4;\mathbb{Z}_N)} \rme^{\frac{2\pi \im }{N} b(\vec{\Sigma} -  \hat{\Sigma}_B) }\right) \notag \\
    &= \sum_{\substack{\vec{\Sigma} \in C_2(M_4; \Lambda_{\mathrm{weights}})  \\ \vec{C} \in C_1(M_4; \Lambda_{\mathrm{roots}}) \\
    \text{satisfying~} \vec{C} = \partial \vec{\Sigma}}} \rme^{-S_{\mathrm{vortex}}[\vec{\Sigma}+\hat{\Sigma}_B] - S_{\mathrm{monopole}}[\vec{C} + \partial \hat{\Sigma}_B]} \left( \sum_{b \in H^2 (M_4;\mathbb{Z}_N)} \rme^{\frac{2\pi \im }{N} b(\vec{\Sigma} ) }\right),
\end{align}
which clearly shows that $\hat{\Sigma}_B$ and $\partial \hat{\Sigma}_B$ serve as the source for the background center-vortex and monopole excitations, respectively. 
This gives another justification of our criteria for the center-vortex condensation (Def.~\ref{def:center-vortex-condensation}) and the monopole condensation (Def.~\ref{def:lens-space-criterion}), which is complementary to the continuum Higgs-phase computation in Sec.~\ref{sec:justification}.
Note that $Z_{\text{model}}[B]$ itself does not depend on the lift of $\Sigma_B$ because the original expression (\ref{eq:center-vortex-model}) is obviously independent. 
In the above representation, one can also check that a change of the lift can be absorbed by a redefinition of $(\vec{\Sigma},\vec{C})$.

Under the assumption that the system is gapped and satisfies center-vortex condensation, the locality yields a nontrivial consequence of the monopole condensation as shown in Claim~\ref{claim:weak-vortex-monopole-criterion} and~\ref{claim:strong-vortex-monopole-criterion}:\footnote{If it is further assumed to be uniquely gapped, the partition function is just an SPT phase: $Z_{\text{model},~M_4}[B] \neq 0$ for any $B$, which obviously implies monopole condensation.} 
$Z_{\text{model},~L(qN;1)\times S^1}[B] \neq 0$ (in the large-volume limit) for some $q>0$.

It would be instructive to examine the situation where the center vortices are proliferated without any monopoles: $S_{\mathrm{vortex}}[\vec{\Sigma}] \rightarrow 0$ and $S_{\mathrm{monopole}}[\vec{C}] \rightarrow (+ \infty) \delta_{\vec{C},0} $.
In this limit, the twisted partition function can be calculated as follows:
\begin{align}
    Z_{\text{model}}[B]= \sum_{\substack{\vec{\Sigma} \in C_2(M_4; \Lambda_{\mathrm{weights}})  \\
    \text{satisfying~} \partial \vec{\Sigma} =0}} \left( \sum_{b \in H^2 (M_4;\mathbb{Z}_N)} \rme^{\frac{2\pi \im }{N} b(\vec{\Sigma}) -  b(\Sigma_B) }\right). 
\end{align}
Remember that $b(\vec{\Sigma})$ should be understood as $b([\vec{\Sigma}])$, where $[\vec{\Sigma}] \in H_2(M_4; \mathbb{Z}_N)$ is obtained by taking modulo $\Lambda_{\mathrm{roots}}$.
We can rewrite the partition function as,
\begin{align}
    Z_{\text{model}}[B]= \sum_{\substack{\vec{\Sigma} \in C_2(M_4; \Lambda_{\mathrm{weights}})  \\
    \text{satisfying~} \partial \vec{\Sigma} =0}} \delta_{H_2(M_4;\mathbb{Z}_N)}
    \left( [\vec{\Sigma}] - \Sigma_B\right),
\end{align}
where $\delta_{H_2(M_4;\mathbb{Z}_N)}(\cdot)$ denotes the Kronecker delta in the $H_2(M_4;\mathbb{Z}_N)$ space.
From this expression, we can absorb $\Sigma_B$ by the shift $\vec{\Sigma} \mapsto \vec{\Sigma} + \hat{\Sigma}_B$ if and only if the Bockstein vanishes $\beta B =0$.
Therefore, we obtain
\begin{align}
    Z_{\text{model}}[B] \propto \delta_{H^3(M_4;\Lambda_{\mathrm{roots}})}(\beta B).
\end{align}
At first sight, this partition function may seem to be the one in a gapped phase without monopole condensation.
However, we should interpret this theory as a limit of the Coulomb phase, rather than a TQFT.
Indeed, the partition function can be rewritten as follows:
\begin{align}
    Z_{\text{model}}[B] &= \sum_{\tilde{a} \in H^1(M_4;\Lambda_{\mathrm{weights}})} \rme^{\im \int_{M_4} \tilde{a} \cup \beta B}= \int_{\text{flat $U(1)^{N-1}$ gauge}} \mathcal{D}\tilde{a}~\rme^{\im \int_{M_4} \tilde{a} \cup \beta B} \notag\\
    &=\lim_{\tilde{g}\to 0}\int \Diff \tilde{a} \exp\left[-\frac{1}{2\tilde{g}^2}\int |\diff \tilde{a}|^2 +\im \int \tilde{a}\cup \beta B\right]. 
\end{align}
In the first line, we have introduced the flat $U(1)^{N-1}$ gauge field, which consists only of the holonomy degrees of freedom, and this gives a Fourier representation of $\delta_{H^3(M_4;\Lambda_{\mathrm{roots}})}(\beta B)$. 
In the second line, we regard it as the weak-coupling limit of the standard Maxwell theory. 
By using the lift $\hat{B}$, we can write $\int_{M_4} a \cup \beta B = \frac{1}{N}\int_{M_4} \diff a \cup \hat{B}$ as $\beta B = \frac{1}{N}\diff \hat{B}$, and thus $B$ couples to $\tilde{a}$ magnetically. Thus, the weak-coupling limit is taken in the magnetic description of the $U(1)^{N-1}$ Maxwell gauge theory.
Via abelian duality, it is the strong-coupling limit of the electric $U(1)^{N-1}$ gauge theory, where $B$ couples electrically.
Therefore, we interpret this partition function as the formal strong-coupling limit of the Coulomb phase: Since the magnetic-field term is dropped from the Hamiltonian in the strong coupling limit, the nonzero magnetic-flux insertion does not cost the energy and thus the center-vortex condensation criterion (Def.~\ref{def:center-vortex-condensation}) is achieved, while it violates the monopole condensation criterion (Def.~\ref{def:lens-space-criterion}). 
Our result is consistent with the recent observation~\cite{Nguyen:2024ikq, Giansiracusa:2025hfj} on the emergence of the Coulomb phase in the $\mathbb{Z}_N$ lattice gauge theories for any $N>2$ by suppressing only monopoles with maintaining center vortices.

Thus, the partition function $Z_{\text{model}}[B]=\delta_{H^3(M_4;\Lambda_{\mathrm{roots}})}(\beta B)$ for the center-vortex ensembles without monopoles should be regarded as the limit of the gapless Coulomb phase rather than a gapped phase.
It requires infinite degrees of freedom even in the low-energy limit, which violates the finiteness assumption (Assumption \ref{assumption:finiteness}), and thus the result is consistent with our Claim~\ref{claim:weak-vortex-monopole-criterion}.

\subsection{Lens-space partition function and higher-charge monopole condensation}
\label{sec:highercharge}

In confinement phases, the condensation of monopoles with minimal magnetic charges is usually assumed, and then the magnetic charge lattice of condensates is identical with the one of all dynamical monopoles, $\Lambda_{\mathrm{cond.}} = \Lambda_{\mathrm{roots}}$. 
Of course, nothing forbids us to consider exotic confinement phases, where the condensate only forms a sublattice of magnetic charges, $\Lambda_{\mathrm{cond.}} \subset \Lambda_{\mathrm{roots}}$, and such examples include oblique confinement phases by dyon condensation with higher magnetic charges~\cite{tHooft:1981bkw, Cardy:1981qy, Cardy:1981fd, Honda:2020txe, Hayashi:2022fkw}.

In this section, let us examine such possibilities based on the center-vortex model~\eqref{eq:center-vortex-model}, and we shall reveal the correspondence between the condensed magnetic charge $\Lambda_{\mathrm{cond.}}$ and the lens-space partition function.
To represent the restriction of the condensed magnetic charges, 
we consider the limit $S_{\mathrm{monopole}}[\vec{C}]=\infty$ if $\vec{C}\not\in C_1(M_4;\Lambda_{\mathrm{cond.}})$ in \eqref{eq:center-vortex-model}:
\begin{align}
    Z_{\text{model}}^{(\Lambda_{\mathrm{cond.}})}[B] &= \sum_{\substack{\vec{\Sigma} \in C_2(M_4; \Lambda_{\mathrm{weights}})  \\ \vec{C} \in C_1(M_4; \Lambda_{\mathrm{cond.}}) \\
    \text{satisfying~} \vec{C} = \partial \vec{\Sigma}}} \rme^{-S_{\mathrm{vortex}}[\vec{\Sigma}] - S_{\mathrm{monopole}}[\vec{C}]} \left( \sum_{b \in H^2 (M_4;\mathbb{Z}_N)} \rme^{\frac{2\pi \im }{N} b(\vec{\Sigma} -  \hat{\Sigma}_B) }\right) \notag\\
    &= \sum_{\substack{\vec{\Sigma} \in C_2(M_4; \Lambda_{\mathrm{weights}})  \\ \vec{C} \in C_1(M_4; \Lambda_{\mathrm{cond.}}) \\
    \text{satisfying~} \vec{C} = \partial \vec{\Sigma}}} \rme^{-S_{\mathrm{vortex}}[\vec{\Sigma}] - S_{\mathrm{monopole}}[\vec{C}]}
    \,\,\delta_{H_2(M_4;\mathbb{Z}_N)}
    \left( [\vec{\Sigma}] - \Sigma_B\right).
\end{align}
We note that the sum over $\vec{C}$ is restricted to $C_1(M_4;\Lambda_{\mathrm{cond.}})\subset C_1(M_4;\Lambda_{\mathrm{roots}})$ compared with \eqref{eq:center-vortex-model}. 
Because of this restriction and the constraint $\vec{C}=\partial\vec{\Sigma}$ of the magnetic flux conservation, 
the twisted boundary condition implies,
\begin{align}
    Z_{\mathrm{model}}^{(\Lambda_{\mathrm{cond.}})}[B] = 0 ~~~\text{if }[\partial \hat{\Sigma}_B] \notin [Z_1(M_4;\Lambda_{\mathrm{cond.}})],
\end{align}
where $Z_1(M_4;\Lambda_{\mathrm{cond.}}) = \{ \vec{C} \in C_1(M_4;\Lambda_{\mathrm{cond.}})| \partial \vec{C}=0 \}$, and $[\cdot]$ stands for the equivalence class in $H_1(M_4;\Lambda_{\mathrm{roots}})$;  $[Z_1(M_4;\Lambda_{\mathrm{cond.}})]=\{[\vec{C}]\in H_1(M_4;\Lambda_{\mathrm{roots}})\,|\, \vec{C}\in Z_1(M_4;\Lambda_{\mathrm{cond.}})\}$.
For simplicity, we here assume the condensation of all (fundamental) center vortices and compute the twisted partition function on $M_4 = L(qN,1) \times S^1$ with the minimal background $B$ satisfying $\beta B\neq 0$ as in Def.~\ref{def:lens-space-criterion}.

We note that $\beta B\in \Lambda_{\mathrm{roots}}/(qN \Lambda_{\mathrm{roots}})=\mathrm{Tor}\, H^3(L(qN,1) \times S^1,\Lambda_{\mathrm{roots}})$, 
%$= \Lambda_{\mathrm{roots}}/(qN \Lambda_{\mathrm{roots}}) \oplus \Lambda_{\mathrm{roots}}$.
%In terms of the Poincar\'e dual, 
and thus $[\partial \hat{\Sigma}_B]\in \Lambda_{\mathrm{roots}}/(qN \Lambda_{\mathrm{roots}})=\mathrm{Tor}\, H_1(L(qN,1) \times S^1,\Lambda_{\mathrm{roots}})$. 
%$ = \Lambda_{\mathrm{roots}}/(qN \Lambda_{\mathrm{roots}}) \oplus \Lambda_{\mathrm{roots}}$.
In the sector $\Lambda_{\mathrm{roots}}/(qN \Lambda_{\mathrm{roots}})$, the condensed magnetic charge lattice $\subset [Z_1(M_4;\Lambda_{\mathrm{cond.}})]$ is given by
\begin{align}
    \Lambda_{\mathrm{cond.}}/(qN \Lambda_{\mathrm{roots}}) := \{ [\vec{\alpha}]_{\operatorname{mod} qN \Lambda_{\mathrm{roots}} } \in \Lambda_{\mathrm{roots}}/(qN \Lambda_{\mathrm{roots}}) ~|~ \vec{\alpha} \in \Lambda_{\mathrm{cond.}} \}.
\end{align}
with the projection $[\cdot]_{\operatorname{mod} qN}:  \Lambda_{\mathrm{roots}} \rightarrow \Lambda_{\mathrm{roots}}/(qN \Lambda_{\mathrm{roots}})$.\footnote{
Note that this definition includes the case when $(qN \Lambda_{\mathrm{roots}})$ is not a sublattice of $\Lambda_{\mathrm{cond.}}$.
For example, when $\Lambda_{\mathrm{cond.}} = 2 \Lambda_{\mathrm{roots}}$ and $Nq =3$, we have $(2\Lambda_{\mathrm{roots}})/(3 \Lambda_{\mathrm{roots}})= \Lambda_{\mathrm{roots}}/(3 \Lambda_{\mathrm{roots}})$.}
On the other hand, if we choose a minimal nontrivial $B$, we can write $[\partial \hat{\Sigma}_B] = q N\vec{\mu}_1 \in \mathrm{Tor}\, H_1(L(qN,1) \times S^1,\Lambda_{\mathrm{roots}})$ (see Appendix \ref{app:Lens_space} for details).
We then obtain that
\begin{align}
    Z_{\mathrm{model}}^{(\Lambda_{\mathrm{cond.}})}[B] = 0 ~~~\text{if }[q N\vec{\mu}_1] \notin \Lambda_{\mathrm{cond.}}/(qN \Lambda_{\mathrm{roots}}). 
    \label{eq:higher_charge_and_lensspace}
\end{align}
Using this formula~\eqref{eq:higher_charge_and_lensspace}, the $q$-dependence of the lens-space twisted partition function tells us the information of the condensed magnetic charge $\Lambda_{\mathrm{cond.}}$ at least partially.

To be more specific, let us consider the case, 
\begin{align}
    \Lambda_{\mathrm{cond.}} = p\Lambda_{\mathrm{roots}},
\end{align}
for some $p>0$, and we see the $q$-dependence of the lens-space  twisted partition function $\{ Z_{L(qN;1)\times S^1}[B] \}_{q\in \mathbb{Z}_{> 0}}$:
% This reveals the correspondence between our definition (Definition \ref{def:lens-space-criterion}) of charge-$[q]_{Nq\mathbb{Z}}$ monopole condensation and the standard ``charge-$p$''monopole condensation $\Lambda_{\mathrm{cond.}} = p\Lambda_{\mathrm{roots}}$.
\begin{itemize}
    \item Let us start with the easiest case, $\gcd(N,p)=1$, where the twisted partition function does not vanish for all $q>0$:
\begin{align}
    Z_{L(qN;1)\times S^1}[B] \rightarrow \text{const.} \neq 0~~~\text{for all }q.
\end{align}
    To show this, let us compute the condensed magnetic-charge lattice modulo $qN\Lambda_{\mathrm{roots}}$:
\begin{align}
    \Lambda_{\mathrm{cond.}}/(qN \Lambda_{\mathrm{roots}}) 
    &= (p\Lambda_{\mathrm{roots}})/(qN \Lambda_{\mathrm{roots}}) \notag\\
    &= (\operatorname{gcd}(qN,p)\Lambda_{\mathrm{roots}})/(qN \Lambda_{\mathrm{roots}})\notag\\
    &=(\operatorname{gcd}(q,p)\Lambda_{\mathrm{roots}})/(qN \Lambda_{\mathrm{roots}}). 
\end{align}    
    We have used the B\'ezout's identity to obtain the second line, and then we note that $\gcd(qN,p)=\gcd(q,p)$ when $\gcd(q,p)=1$ for the third line. 
Since $(N\vec{\mu}_1) \in \Lambda_{\mathrm{roots}}$, the magnetic charge imposed from the twisted boundary condition stays within the condensed magnetic charge:
\begin{align}
    [q N\vec{\mu}_1] \in \Lambda_{\mathrm{cond.}}/(qN \Lambda_{\mathrm{roots}}),~~~\text{for all }q,~\text{if }\operatorname{gcd}(N,p)=1. 
\end{align}

    Hence, from the twisted partition function, we cannot distinguish the condensates of  $\Lambda_{\mathrm{cond.}} = p\Lambda_{\mathrm{roots}}$ and $\Lambda_{\mathrm{cond.}} = \Lambda_{\mathrm{roots}}$ if $\operatorname{gcd}(N,p)=1$.
    This result would be reasonable because we only use $\mathbb{Z}_N^{[1]}$ symmetry for the twisted boundary condition.

    \item For $\operatorname{gcd}(N,p) \neq 1$, the twisted partition function would exhibit the following pattern for minimal $B=A_{\mathbb{Z}_N}\cup \delta(x_4)\diff x_4$:
\begin{align}
    Z_{L(qN;1)\times S^1}[B] \rightarrow 
    \begin{cases}
        \neq 0   ~~~&(q=0 ~\operatorname{mod}C_{N,p})\\
        =0   ~~~&(q\neq 0 ~\operatorname{mod}C_{N,p}),
    \end{cases}
\end{align}
where the integer $C_{N,p}$ is the maximal divisor of $p$ whose prime factors divide $N$, i.e. $C_{N,p} = \operatorname{gcd}(N^K,p)$ with $K\gg 1$.
We can prove this statement as follows:
\begin{align}
    Z_{L(qN;1)\times S^1}[B]\not\to 0 
    \Leftrightarrow 
    &~ [q N\vec{\mu}_1] \in \Lambda_{\mathrm{cond.}}/(qN \Lambda_{\mathrm{roots}})=(\operatorname{gcd}(Nq,p)\Lambda_{\mathrm{roots}})/(qN \Lambda_{\mathrm{roots}})\notag\\
    \Leftrightarrow
    &~ \gcd(Nq,p)| q \notag\\
    \Leftrightarrow
    &~ C_{N,p}|q. 
\end{align}
The equivalence on the last line can be obtained by considering the prime decomposition of $N,p,q$ and comparing the exponent of each prime factor.

Thus, the $q$-dependence of the lens-space partition function detects $C_{N,p}=\gcd(N^K,p)$.
The information contained in $p$ is preserved for the prime factors of $N$, but lost for the factors of $p$ that do not divide $N$.
For $\operatorname{gcd}(N,p)=1$, all information of $p$ is lost.
\end{itemize}

In the latter case, we have a theory with the following properties:
\begin{align}
    Z_{T^4}[B] &\rightarrow \text{const.} \neq 0, \\
    Z_{L(qN;1)\times S^1}[B] &\rightarrow 
    \begin{cases}
        \neq 0   ~~~&(q=0 ~(\operatorname{mod}C_{N,p}))\\
        =0   ~~~&(q\neq 0 ~(\operatorname{mod}C_{N,p})),
    \end{cases}
\end{align}
Let us emphasize that this pattern does not exist in the conventional Wilson-'t Hooft classification~\cite{tHooft:1979rtg}.
Usually, in the exotic gapped phases with nontrivial magnetic charges, we suppose that the dyon condensation pattern automatically fixes the center-vortex condensation pattern.
However, in the present example, the center vortex is fully condensed, whereas the magnetically charged object shows the nontrivial condensation pattern.
Thus the higher-charge condensation $\Lambda_{\mathrm{cond.}}=p\Lambda_{\mathrm{roots}}$ with $\gcd(N,p)\not=1$ gives the gapped phase beyond the conventional classification.

Its TQFT description requires us to consider the symmetry fractionalization of the $\mathbb{Z}_N^{[1]}$ symmetry, which is given by the level-$p$ BF theory with the magnetic coupling,
\begin{align}
    Z_{M_4}^{\mathrm{BF}_p}[B] 
    &= \int \mathcal{D}b \mathcal{D}a ~\exp\left({\frac{\im p}{2\pi}  \int_{M_4} da \cup b +  \im \int_{M_4} a \cup \beta_{\mathbb{Z}} B}\right) \notag\\
    &= \sum_{a \in H^1(M_4;\mathbb{Z}_p)} \exp\left({ \frac{2 \pi \im}{p} \int_{M_4} a \cup \beta_{\mathbb{Z}} B}\right). 
\end{align}
We here use the standard geometric normalization for $U(1)$ $1$- and $2$-form gauge fields, $a$, $b$, for the first line, 
and we integrate out $b$ to obtain the second line with changing the convention as $a\to \frac{2\pi}{p}a$ so that $a\in H^1(M_4,\mathbb{Z}_p)$. 
% where $\beta_{\mathbb{Z}}: H^2(M_4;\mathbb{Z}_N) \rightarrow H^3(M_4;\mathbb{Z})$ is the Bockstein of $0 \rightarrow \mathbb{Z} \xrightarrow{\times N}  \mathbb{Z} \xrightarrow{\operatorname{mod} N}   \mathbb{Z}_N  \rightarrow 0$.
% Indeed, $H^3(L(qN;1)\times S^1;\mathbb{Z}) = \mathbb{Z} \oplus \mathbb{Z}_{qN}$, and the Bockstein map $\beta_{\mathbb{Z}}$ is the homomorphism from $ \mathbb{Z}_N \subset H^2(L(qN;1)\times S^1;\mathbb{Z}_N) $ to $\mathbb{Z}_{qN} \subset H^3(M_4;\mathbb{Z})$.
For our choice of $B$, we can write $\beta_{\mathbb{Z}} B=q \in \mathbb{Z}_{qN} \subset H^3(M_4;\mathbb{Z})$, which leads to
%In the level-$p$ BF theory, this leads to
\begin{align}
    Z_{L(qN;1)\times S^1}^{\mathrm{BF}_p}[B] = \delta_{\operatorname{mod} \operatorname{gcd}(Nq,p)}(q)    \begin{cases}
        \neq 0   ~~~&(q=0 ~(\operatorname{mod}C_{N,p}))\\
        =0   ~~~&(q\neq 0 ~(\operatorname{mod}C_{N,p})).
    \end{cases}
\end{align}
Thus, the $q$-dependence of the lens-space twisted partition function is correctly reproduced. 
In this theory, while the center vortex is fully condensed, some nontrivial selection rule exists for the condensing monopoles.
This is the situation where the magnetic symmetry is emergent and spontaneously broken, leading to a nontrivial TQFT while preserving the original electric $\mathbb{Z}_N^{[1]}$ symmetry. This is the mechanism to obtain the gapped phase beyond the conventional Wilson-'t Hooft classification in this center-vortex framework.

Lastly, let us briefly comment on the case when the proliferated center-vortex flux is also restricted. 
So far, we have assumed the center symmetry is fully preserved and the 't~Hooft flux $B$ is chosen to be a minimal one.
However, when the condensed magnetic charge is nontrivial as $\Lambda_{\mathrm{cond.}} = p\Lambda_{\mathrm{roots}}$, it is often the case that the center symmetry $\mathbb{Z}_N^{[1]}$ is partially broken to $\mathbb{Z}_{N/\operatorname{gcd}(N,p)}^{[1]}$, and then the vacuum fits into the conventional Wilson-'t~Hooft classification. 
In the center-vortex model~\eqref{eq:center-vortex-model}, this can be achieved by requiring the condensed center vortex carries the restricted magnetic fluxes: 
\begin{align}
    \vec{\Sigma}\in C_2(M_4; \gcd(N,p)\Lambda_{\mathrm{weights}}). 
\end{align}
Then, Wilson loops $W^{N/\gcd(N,p)}(C)$ do not rotate their phases by the restricted center vortices and thus they are naturally deconfined: $\mathbb{Z}_N^{[1]}\xrightarrow{\mathrm{SSB}}\mathbb{Z}_{N/\gcd(N,p)}^{[1]}$. 
The minimal magnetic charge forced by the twisted boundary condition is $\operatorname{gcd}(N,p) q N\vec{\mu}_1$.
We can see that $\operatorname{gcd}(N,p) q N\vec{\mu}_1 \in (\operatorname{gcd}(Nq,p)\Lambda_{\mathrm{roots}})/(qN \Lambda_{\mathrm{roots}})$ for all $q$.
Hence, the spontaneous breaking $\mathbb{Z}_N^{[1]} \rightarrow \mathbb{Z}_{N/\operatorname{gcd}(N,p)}^{[1]}$ with at most SPT stacking, i.e., $Z[B] \propto \delta(\operatorname{gcd}(N,p) B)$, is consistent with charge-$p$ monopole condensation, without requiring any extra condition on the Bockstein element.

\section{Summary and Discussion}
\label{sec:discussion}

\paragraph{Summary:}
In this paper, we give several gauge-invariant criteria for the center-vortex and monopole condensation using the $\mathbb{Z}_N^{[1]}$-twisted partition function:
Weak and strong center-vortex condensations (Defs.~\ref{def:center-vortex-condensation} and~\ref{def:center-vortex-condensation-strong}, respectively) are defined by the torus twisted partition function, and weak and strong monopole condensations (Def.~\ref{def:lens-space-criterion} and~\ref{def:lens-space-criterion-strong}, respectively) are defined by the lens-space twisted partition function.
After explaining why these are sensible definitions by computing them on the adjoint Higgs phase, we discuss the relation between these condensation criteria:
For the gapped phases with locality and finiteness, we have proven the equivalence between weak center-vortex condensation and weak monopole condensation.
If we further assume the remote detectability in the low-energy TQFT, we can also show the equivalence between the strong monopole condensation and the unbroken $\mathbb{Z}_N^{[1]}$ symmetry.
The outcome can be summarized as follows:
\begin{align}
    &\text{Unique gapped vacuum} \notag\\
    \Rightarrow &~
    \text{Def.}~\ref{def:center-vortex-condensation-strong}~~\text{(strong  center-vortex condensation)} \notag\\
    \Rightarrow
    & 
    \left\{\begin{array}{l}
        \text{Unbroken $\mathbb{Z}_N^{[1]}$ symmetry} ~~\Leftrightarrow \\
        \text{Def.}~\ref{def:lens-space-criterion-strong}~~\text{(strong charge-$[q]_{Nq\mathbb{Z}}$ monopole condensation for some $q>0$)}  
    \end{array}\right.\notag\\
    \Rightarrow& 
    \left\{\begin{array}{l}
        \text{Def.}~\ref{def:center-vortex-condensation}~~ \text{(weak center-vortex condensation)}~~\Leftrightarrow\\
         \text{Def.}~\ref{def:lens-space-criterion}~~\text{(weak charge-$[q]_{Nq\mathbb{Z}}$ monopole condensation for some $q>0$)}.
    \end{array}\right.  \tag{\ref{eq:summary_relation}}
\end{align}

As an illustrative example to understand the results, we consider a formal-sum representation of the center-vortex percolating model with monopole worldlines.
This example indicates that the system characterized by center-vortex condensation in the absence of monopole condensation can be regarded as a limit of the Coulomb phase, which violates the finiteness assumption as a gapped phase.
This result is consistent with the recent observation~\cite{Nguyen:2024ikq, Giansiracusa:2025hfj} by Tin Sulejmanpasic and his collaborators.
We also discuss the case of higher-magnetic-charge condensation using the center-vortex percolating model, and we obtain a nontrivial $q$-dependence in the lens-space twisted partition function.
It is worthwhile to note that such a phase is beyond the conventional Wilson-'t~Hooft classification, which is described by TQFTs with the nontrivial symmetry fractionalization of the $\mathbb{Z}_N^{[1]}$ symmetry.
%the correspondence between our charge-$[q]{Nq\mathbb{Z}}$ monopole condensation and the standard ``charge-$p$’'monopole condensation $\Lambda_{\mathrm{cond.}} = p\Lambda_{\mathrm{roots}}$.
% It is worthwhile to note that some gapped phases with nontrivial charge-$[q]_{Nq\mathbb{Z}}$ monopole condensation exist, which is beyond the conventional Wilson-'t~Hooft classification.

\vspace{0.5em}
In what follows, let us comment on several related topics in the discussion.

\paragraph{Deconfined but weak center-vortex/monopole condensed phase:}

As indicated in the relation~\eqref{eq:summary_relation} for the condensation criteria, the unbroken $\mathbb{Z}_N^{[1]}$ symmetry implies the weak center-vortex/monopole condensation, but the converse does not necessarily hold.

Let us provide an example that satisfies the weak center-vortex condensation criterion but supports the deconfined $\mathbb{Z}_N^{[1]}$-charged lines.
Let us consider the discrete gauge theory with the finite Heisenberg gauge group $H_N$, which is generated by the $SU(N)$ clock and shift matrices $C$ and $S$, i.e. $H_N=\langle S,C \,|\, S^N=C^N=(-1)^{N-1}\bm{1}_N,  SC=\rme^{2\pi\im/N}CS\rangle$. For $N=2$, it is the quaternion group $Q_8$.
This gauge group fits into the short exact sequence,
\begin{align}
1\to \mathbb{Z}_N\to H_N\to \mathbb{Z}_N\times \mathbb{Z}_N\to 1,
\end{align}
so it has the $\mathbb{Z}_N^{[1]}$ symmetry as the center symmetry, and the Wilson lines are deconfined charged objects.
As the 't Hooft flux on $T^2$ can be saturated by a flat configuration of this discrete gauge group, the $\mathbb{Z}_N^{[1]}$-twisted partition function does not vanish: On $T^4$, we can parametrize the holonomy along each cycle as $g_i=\rme^{2\pi \im n_i/N}C^{a_i}S^{b_i}\in H_N$, which satisfy the commutation relation $g_ig_j = \rme^{\frac{2\pi\im }{N}(a_ib_j-b_i a_j)}g_j g_i$.
Since this $\mathbb{Z}_N$ phase should be identical with the background gauge field $B$, we find (here $(a\wedge b)_{ij}:=a_i b_j-b_i a_j$)
\begin{equation}
    \frac{Z_{T^4}^{(H_N)}[B]}{Z^{(H_N)}_{T^4}[0]}=\frac{|\{(a,b)\in \mathbb{Z}_N^4\times \mathbb{Z}_N^4\,|\, a\wedge b=B\}|}{|\{(a,b)\in \mathbb{Z}_N^4\times \mathbb{Z}_N^4\,|\, a\wedge b=0\}|}.
\end{equation}
For $B_{34}\not =0$ with other $B_{ij}=0$, this ratio is nonzero, and the weak center-vortex condensation occurs (so does the weak monopole condensation), even though the Wilson loops are deconfined.
Whereas discrete gauge theories may have a center vortex along the global cycles without action cost, this fact does not always imply the unbroken $\mathbb{Z}_N^{[1]}$ symmetry.

\paragraph{$\theta$ dependence, anomaly, and monopole}

In this paper, we have shown the monopole condensation from the center-vortex condensation with the assumption of the mass gap and locality.
For the case of the $SU(N)$ Yang-Mills theory with adjoint matter, there exists the $\theta$ parameter, and the existence of monopoles in the confinement vacua can also be expected from the $\theta$ dependence.
We note that $\theta$ is a $2\pi$ periodic parameter, but it has the mixed 't~Hooft anomaly with the $\mathbb{Z}_N^{[1]}$ symmetry~\cite{Gaiotto:2017yup},
\begin{align}
    Z_{\theta+ 2\pi}[B] = \rme^{\frac{2 \pi \im}{N} \int_{M_4} \frac{1}{2} P_2(B)} Z_{\theta}[B].
\label{eq:theta_anomaly}
\end{align}
This anomaly cannot be matched by a single TQFT description without local degeneracy unless $Z[B]=0$ for $P_2(B)\not=0$, because finite TQFTs cannot have nontrivial continuous parameter dependence.
Since $H_N\subset SU(N)$, we note that the above $H_N$ gauge theory should satisfy the anomaly relation as it is the low-energy effective theory for the Higgsing $SU(N)\xrightarrow{\mathrm{Higgs}}H_N$, and it is actually the case with setting $Z_{\theta}[B]=Z^{(H_N)}[B]$ because $Z^{(H_N)}[B]=0$ when $B\cup B\not=0$.

Assuming the unbroken $\mathbb{Z}_N^{[1]}$ symmetry,
this anomaly~\eqref{eq:theta_anomaly} implies at least some nontrivial $\theta$ dependence, typically leading to the multi-branch structure~\cite{Witten:1980sp, DiVecchia:1980yfw}.
This, in particular, requires the nontrivial $\theta$ dependence of the ground state energy.
If we further assume the Abelianization $SU(N) \rightarrow U(1)^{N-1}$, the topological charge can be written in the following form;
\begin{align}
    Q_{\mathrm{top}} = \frac{1}{8\pi^2}\int_{M_4} \vec{f} \wedge \vec{f},
\end{align}
where $\vec{f} $ denotes the $U(1)^{N-1}$ field strength.
If monopoles are totally absent in the Abelianized description, $Q_{\mathrm{top}}=0$ on $S^4$ and thus the ground-state energy density is independent of $\theta$.
To be consistent with the anomaly relation~\eqref{eq:theta_anomaly}, the existence of monopoles is necessary for producing a nontrivial $\theta$ dependence for the $\mathbb{Z}_N^{[1]}$-unbroken gapped phase.
Although this argument is not strong enough to imply the ``condensation’’ of monopoles, their necessity can be understood in a simpler way.

The same argument applies to the center-vortex picture~\cite{Engelhardt:1999xw, Reinhardt:2001kf, Cornwall:1999xw}.
Because $\vec{\Sigma} $ is nothing but the magnetic flux of the $U(1)^{N-1}$ gauge group,  the topological charge in this model is given by
\begin{align}
    Q_{\mathrm{top}} = \operatorname{Intersection} (\vec{\Sigma} , \vec{\Sigma}).
    \label{eq:Qtop_vortex_model}
\end{align}
The non-zero topological charge arises from two contributions: (1) the intersection of global fluxes, and (2) the linking between the monopole worldlines and the magnetic fluxes.\footnote{Note that a small instanton can be realized by a center vortex winding around a monopole worldline. Such a localized center-vortex worldsheet cannot be removed by a continuous deformation, since the monopole constitutes a singularity where the abelianization breaks down.}
Because there should be local objects with non-zero topological charge for the nontrivial $\theta$ dependence, the monopoles are necessary in the vacuum configurations.

At the local intersection of center-vortex worldsheets, the topological charge density~\eqref{eq:Qtop_vortex_model} is given by the inner product of weight vectors for those magnetic fluxes, which is quantized in $1/N$.
The local fractional feature of the topological charge density provides a natural scenario for the multi-branch $\theta$ vacua~\cite{Cornwall:1999xw}:
The total topological charge is quantized to integers, but the local constituents are fractionalized and independently liberated.
This mechanism for the multi-branch structure is explicitly realized in the semiclassical descriptions of confinement in the compactified setups~\cite{Tanizaki:2022ngt, Tanizaki:2022plm, Hayashi:2023wwi, Hayashi:2024gxv, Hayashi:2024qkm, Hayashi:2024psa, Hayashi:2025doq, Yamazaki:2017ulc, Cox:2021vsa, Unsal:2007vu, Unsal:2007jx, Unsal:2008ch, Shifman:2008ja, Poppitz:2012nz, Poppitz:2012sw, Davies:2000nw}.
Actually, the tight connection between monopoles and center vortices in the semiclassical setup~\cite{Hayashi:2024yjc} provided the authors with the strong motivation to revisit their relation, which triggered this study.

\paragraph{Future outlooks:}

Lastly, let us mention miscellaneous topics as possible future directions.
The leading motivation for our proposal is to enforce the solitonic excitations using the symmetry-twisted boundary condition in the symmetry-broken phases.
In the context of the generalized symmetry, the representation theory for solitons has been developed in Ref.~\cite{Gagliano:2025gwr, Cordova:2024iti},
and it would be intriguing to explore the representation theory of center vortices and monopoles through the lens of twisted partition functions.

We also note that the computation of the partition function on $L(p,1)\times S^1$ has appeared in the context of the superconformal index for $4$d gauge theories~\cite{Benini:2011nc, Razamat:2013jxa,  Razamat:2013opa}.
To the best of our knowledge, our proposal is the first one to interpret the lens-space twisted partition function as the criterion for the monopole condensation in $(3+1)$d, but it is definitely an interesting question to revisit these computations of the supersymmetric gauge theories to learn if we can obtain a new insight into the monopoles and dyons.

According to~\eqref{eq:lens_partition_function_trace},
the lens-space twisted partition function counts the number of ground states on $\mathcal{H}_0(L(qN;1))$ weighted with the $1$-form symmetry charge.
In a recent work~\cite{Anber:2026pfg}, Mohamed Anber revisited the path integral of the Maxwell theory on the asymptotically locally Euclidean (ALE) spaces~\cite{Bianchi:1996zj}.
The key observation there is that the $A$-type ALE space has the lens-space asymptotic boundary, and the path integral on the ALE space should be understood as the state vector on the lens space, which clarifies the meaning of the nontrivial transformation properties under the electromagnetic duality transformations.
It is curious if the computation of TQFTs with the $\mathbb{Z}_N^{[1]}$ symmetry on the ALE spaces makes the fruitful connection between our studies and these results.

As mentioned in the main text, the $(3+1)$d TQFTs with the remote detectability are strongly constrained, and they are described by the discrete gauge theories.
This allows us to perform the computation of the twisted partition functions in a concrete manner for various manifolds, and it would be useful to investigate the connection between our argument and the classification of gapped phases through those computations.
In particular, we should revisit the validity of the conventional wisdom in the dual superconductor picture and investigate possible counterexamples and/or extensions from the modern viewpoints.

\acknowledgments

This work was initiated by the conversation when one of the authors (Y.T.) was preparing the lecture material for the SNU-APCTP winter school at Seoul National University, and we thank the organizers for giving us the opportunity triggering this research. 
To complete this work, the discussion during the Yukawa Institute workshops, ``Interfaces \& Symmetry'' (YITP-I-25-04), YITP International School on EIC Physics (YITP-W-25-16), and  ``Buenas Ideas on the QCD Phase Diagram'' (YITP-26-4), were quite valuable. 
This work was partially supported by Japan Society for the Promotion of Science (JSPS)
Research Fellowship for Young Scientists Grant No. 23KJ1161 (Y.H.), by JSPS KAKENHI
Grant No. 26K07106 (Y.T.), and by Center for Gravitational Physics and Quantum Information (CGPQI) at Yukawa Institute for Theoretical Physics. 

\appendix

\section{Technical details on the lens space \texorpdfstring{$L(M;1)$}{L(M;1)} and Bockstein map}
\label{app:Lens_space}

In this appendix, we summarize the topological properties of $L(M;1)$ and $L(M;1) \times S^1$. We particularly utilize the case when $M$ is a multiple of $N$ in the main text.

\subsection{Lens space}
The lens space $L(M;1) $ can be constructed by $S^3/\mathbb{Z}_M$. 
To be more specific, let us regard $S^3=\{(z,w)\in \mathbb{C}^2 \;|\; |z|^2+|w|^2=1\}$ and we take the $\mathbb{Z}_M$ action $\rho:(z,w)\mapsto (\rme^{2\pi\im/M}z, \rme^{2\pi \im/M}w)$. 
As this is the free action, we can define a smooth $3$-manifold by taking its quotient, i.e.
\begin{align}
    L(M; 1) &:= \left\{ (z, w) \in \mathbb{C}^2 \;\middle|\; |z|^2 + |w|^2 = 1 \right\} \Big/ \sim \nonumber \\
    &\qquad \text{where} \quad (z, w) \sim \left( e^{2\pi i / M}z, e^{2\pi i / M}w \right) .
\end{align}
For the later purpose, it is convenient to describe its CW-decomposition structure. We first take the cellular decomposition of $S^3$, where the $3,2,1,0$-cells $e^3_k, e^2_k, e^1_k$ and $e^0_k$ with $k=0,1,\ldots,M-1$ are given by (see \cite[Example~2.43]{Hatcher:0}): 
\begin{align}
    e^3_k&=\left\{ \frac{2\pi k}{M}< \arg(z) < \frac{2\pi (k+1)}{M}, \, 0<|z|\le 1\right\}, \notag\\
    e^2_k&=\left\{\arg(z)=\frac{2\pi k}{M},\, 0<|z|\le 1\right\}, \notag\\
    e^1_k&= \left\{z=0,\, \frac{2\pi k}{M}< \arg(w) < \frac{2\pi (k+1)}{M}\right\}, \notag\\
    e^0_k&=\{(0,\rme^{2\pi \im k/M})\}. 
\end{align}
The $\mathbb{Z}_M$ group action is $\rho(e^\bullet_k)=e^\bullet_{k+1}$, and thus we identify $L(M;1)=e^3_0\cup e^2_0\cup e^1_0\cup e^0_0$ as the CW complex. 
Equivalently, a lens space can be topologically constructed by gluing two solid tori along their boundary surfaces via a Dehn twist: In the above cell decomposition, $e^1_0$ should be regarded as a core $1$-cycle of a solid torus, $e^2_0$ gives the meridional disk of degree $M$, and $e^3_0$ fills the remaining region.

The homology and cohomology are listed as follows:
\begin{align}
    H^k(L(M;1), \mathbb{Z}) = H_{3-k}(L(M;1), \mathbb{Z}) =
    \begin{cases}
        \mathbb{Z}~~~&(k=0) \\
        0~~~&(k=1) \\
        \mathbb{Z}_M~~~&(k=2) \\
        \mathbb{Z}~~~&(k=3) \\
    \end{cases}
\end{align}
For the homology, the result immediately follows from the above cell decomposition: The boundary operator acts as $\partial e^3_0=e^2_1-e^2_0=0$, $\partial e^2_0=e^1_0+\cdots+e^1_{M-1}=Me^1_0$, and $\partial e^1_0=e^0_1-e^0_0=0$. Thus, the chain complex for the $\mathbb{Z}$-valued homology is $\mathbb{Z}\xrightarrow{\times 0}\mathbb{Z}\xrightarrow{\times M}\mathbb{Z}\xrightarrow{\times 0}\mathbb{Z}\to 0$, which gives the result. The cohomology follows from the Poincare duality. 

For the $\mathbb{Z}_N$ coefficient, the chain complex becomes $\mathbb{Z}_N\xrightarrow{\times 0}\mathbb{Z}_N\xrightarrow{\times M}\mathbb{Z}_N\xrightarrow{\times 0}\mathbb{Z}_N\to 0$, and the result becomes
\begin{align}
    H^k(L(M;1), \mathbb{Z}_N) = H_{3-k}(L(M;1), \mathbb{Z}_N)=
    \begin{cases}
        \mathbb{Z}_N~~~&(k=0) \\
        \mathbb{Z}_{\gcd(M,N)}~~~&(k=1) \\
        \mathbb{Z}_{\gcd(M,N)}~~~&(k=2) \\
        \mathbb{Z}_N~~~&(k=3) \\
    \end{cases}
\end{align}
One can obtain the same result by applying the universal coefficient theorem to the result of the $\mathbb{Z}$-valued (co)homology. 
For the $k=1,2$ cases, it is sometimes useful to take care of the natural normalizations of generators, and they are given by $H_1(L(M;1),\mathbb{Z}_N)=\{k[e^1_0]\;|\; k\in \mathbb{Z}_{\gcd(M,N)}\}$ and $H_2(L(M;1),\mathbb{Z}_N)=\{\frac{N k}{\gcd(N,M)}[e^2_0] \;|\; k\in \mathbb{Z}_{\gcd(M,N)}\}$, respectively. 
In particular, for $M=qN$, $\partial e^2_0=qN e^1_0=0$ mod $N$, and thus $[e^2_0]$ generates the nontrivial $2$-cycle of $L(qN;1)$ for the $\mathbb{Z}_N$-valued homology. 

Let us discuss the Bockstein map associated with $0 \rightarrow \mathbb{Z} \xrightarrow{\times N}  \mathbb{Z} \xrightarrow{\operatorname{mod} N}   \mathbb{Z}_N  \rightarrow 0$, 
\begin{align}
\beta_{ \mathbb{Z}}: H^1(L(M;1), \mathbb{Z}_N) \rightarrow H^2(L(M;1), \mathbb{Z}). 
\end{align}
This is the connecting homomorphism of the long-exact sequence, 
\begin{align}
   \cdots\xrightarrow{\times N} \underbrace{H^1(L(M;1))}_{=0}\xrightarrow{\bmod N} \underbrace{H^1(L(M;1),\mathbb{Z}_N)}_{=\mathbb{Z}_{\gcd(M,N)}}\xrightarrow{\beta_{\mathbb{Z}}}\underbrace{H^2(L(M;1))}_{=\mathbb{Z}_M}\xrightarrow{\times N} \cdots, 
\end{align}
%where the left-hand side is isomorphic to $\mathbb{Z}_{\gcd(M,N)}$ and the right-hand side is isomorphic to $\mathbb{Z}_M$. 
which shows the map $\beta_{\mathbb{Z}}$ is injective for this case. % and maps isomorphically onto the $\mathbb{Z}_{\gcd(M,N)}$ subgroup of $H^2(L(M;1), \mathbb{Z})$.
Thus, $\beta_{\mathbb{Z}}$ is nontrivial as long as $\gcd(M, N)$ is nontrivial. 
We can give the concrete expression for the homology version $\beta_{\mathbb{Z}}:H_2(L(M;1), \mathbb{Z}_N)\to H_1(L(M;1))$ as 
\begin{align}
    \beta_{\mathbb{Z}}\left(\frac{N}{\gcd(N,M)}[e^2_0]\right)= \frac{M}{\gcd(N,M)}[e^1_0], 
    \label{eq:derivaton_BocksteinHom}
\end{align}
i.e. ``$\beta_{\mathbb{Z}}=\frac{1}{N}\partial$'', and the cohomology version is obtained by taking the Poincare dual of this correspondence. Especially when $M=qN$, we obtain $\beta_{\mathbb{Z}}([e^2_0])=q[e^1_0]$. 
%(Kunneth; weight/root lattice is basically )

For the arguments in the main text, we also use the cohomology of $L(M;1)\times S^1$ so let us summarize its result, which can be obtained by the K\"unneth formula. 
With the coefficients of the root/weight lattice $\Lambda_* \simeq \mathbb{Z}^{N-1}, (*=\mathrm{roots,~weights})$,
\begin{align}
    H^k(L(M;1)\times S^1, \Lambda_*) =
    \begin{cases}
        \Lambda_*~~~&(k=0) \\
        \Lambda_*~~~&(k=1) \\
        \Lambda_*/(M\Lambda_*)~~~&(k=2) \\
        \Lambda_*/(M\Lambda_*)\oplus \Lambda_*~~~&(k=3) \\
        \Lambda_*~~~&(k=4) 
    \end{cases}
\end{align}
and, with $\mathbb{Z}_N$ coefficient,
\begin{align}
    H^k(L(M;1)\times S^1, \mathbb{Z}_N) =
    \begin{cases}
        \mathbb{Z}_N~~~&(k=0) \\
        \mathbb{Z}_{\gcd(M,N)}\oplus \mathbb{Z}_N~~~&(k=1) \\
        \mathbb{Z}_{\gcd(M,N)}\oplus \mathbb{Z}_{\gcd(M,N)}~~~&(k=2) \\
        \mathbb{Z}_N\oplus \mathbb{Z}_{\gcd(M,N)}~~~&(k=3) \\
        \mathbb{Z}_N~~~&(k=4) 
    \end{cases}
\end{align}

\subsection{More details on the Bockstein map}

We have already discussed the Bockstein homomorphism on the lens space for the short exact sequence, $0\to \mathbb{Z}\to \mathbb{Z}\to \mathbb{Z}_N\to 0$. 
We here discuss the other Bockstein homomorphism that is used for the center-vortex model: Our interest is the lift from $\mathbb{Z}_N$ to $\Lambda_{\mathrm{weights}}$, characterized by the following sequence
\begin{align}
   0 \rightarrow \Lambda_{\mathrm{roots}} \rightarrow \Lambda_{\mathrm{weights}} \xrightarrow{\text{N-ality}} \mathbb{Z}_N \rightarrow 0.
\end{align}
As our main question is whether or not $\mathbb{Z}_N $ 1-form background field can be lifted and embedded into $U(1)^{N-1}$ gauge field (without violating Bianchi identity),  
we focus on the case $M=qN$, i.e. the lens space $L(qN;1)\times S^1$, 
and we compute the Bockstein map 
\begin{align}
    \beta_{\Lambda_{\mathrm{roots}}}: H^2(L(qN;1)\times S^1, \mathbb{Z}_N)  \rightarrow H^3(L(qN;1)\times S^1, \Lambda_{\mathrm{roots}}).
\end{align}
Due to the torsion part $\Lambda_{\mathrm{roots}}/(qN\Lambda_{\mathrm{roots}})$ in $H^3(L(qN;1)\times S^1, \Lambda_{\mathrm{roots}})$, this map can be nontrivial, and we construct it explicitly using $\beta_{ \mathbb{Z}}: H^1(L(qN;1), \mathbb{Z}_N) \rightarrow H^2(L(qN;1), \mathbb{Z})$.

Let $dx_4/L_4$ denote the generator of $H^1(S^1, \mathbb{Z})$. We use the same notation for its pullback to $H^*(L(qN;1)\times S^1, \mathbb{Z})$, and retain this notation even when changing the coefficient ring appropriately via tensor products.
Similarly, we define the generator $A_{\mathbb{Z}_N} \in H^1 (L(qN;1), \mathbb{Z}_N)$ and its image $\beta_{\mathbb{Z}} (A_{\mathbb{Z}_N}) \in H^2 (L(qN;1), \mathbb{Z})$.
We should note that, while $H^2 (L(qN;1), \mathbb{Z}) \simeq \mathbb{Z}_{qN}$, the image should obey $N \beta_{\mathbb{Z}} (A_{\mathbb{Z}_N})  = 0$ (on $\mathbb{Z}_{qN}$), because the map $\beta_{\mathbb{Z}}$ is a homomorphism from $\mathbb{Z}_N$ to $\mathbb{Z}_{qN}$. 
We have already explicitly confirmed it in \eqref{eq:derivaton_BocksteinHom} on the homology side, and this is just its translation on the cohomology side. 

Now, we can explicitly write down the nontrivial part of the Bockstein map as
\begin{align}
   \beta_{\Lambda_{\mathrm{roots}}}: H^2(L(qN;1)\times S^1, \mathbb{Z}_N)  &\rightarrow H^3(L(qN;1)\times S^1, \Lambda_{\mathrm{roots}}) \notag \\
   A_{\mathbb{Z}_N} \cup (dx_4/L_4)~~ &\mapsto ~~ (N \vec{\mu}_1) \beta_{\mathbb{Z}} (A_{\mathbb{Z}_N}) \cup (dx_4/L_4). 
\end{align}
Here, we express the lift from $\mathbb{Z}_N$ to $\Lambda_{\mathrm{weights}}$ as $k \mapsto k \vec{\mu}_1$.
One might worry that this map depends on the choice of lift, but it is actually well-defined.
Indeed, in any other lift, we replace $\vec{\mu}_1$ with $\vec{\mu}_1 + \vec{\alpha}$ for some root vector $ \vec{\alpha}$, 
and then the image of the Bockstein map will be modified as
$(N (\vec{\mu}_1+\vec{\alpha})) \beta_{\mathbb{Z}} (A_{\mathbb{Z}_N}) \cup (dx_4/L_4) = N \vec{\mu}_1 \beta_{\mathbb{Z}} (A_{\mathbb{Z}_N}) \cup (dx_4/L_4) + N \vec{\alpha} \beta_{\mathbb{Z}} (A_{\mathbb{Z}_N}) \cup (dx_4/L_4)$, but the latter term vanishes because it takes a value in the $\Lambda_{\mathrm{roots}}/(Nq\Lambda_{\mathrm{roots}})$, and $N \beta_{\mathbb{Z}} (A_{\mathbb{Z}_N})  = 0~(\operatorname{mod}Nq)$.
In other words, under the identification $H^3(L(qN;1)\times S^1, \Lambda_{\mathrm{roots}}) = \Lambda_{\mathrm{roots}}/(Nq \Lambda_{\mathrm{roots}}) \oplus \Lambda_{\mathrm{roots}}$, we can regard $\beta_{\Lambda_{\mathrm{roots}}}(A_{\mathbb{Z}_N} \cup (dx_4/L_4)) = q N \vec{\mu}_1 \in \Lambda_{\mathrm{roots}}/(Nq \Lambda_{\mathrm{roots}})$, where we set $\beta_{\mathbb{Z}} (A_{\mathbb{Z}_N}) = q \in \mathbb{Z}_{qN} \subset H^{2}(L(qN;1),\mathbb{Z})$ by properly choosing a generator of $\mathbb{Z}_{qN} = H^{2}(L(qN;1),\mathbb{Z})$ as we have done for the homology computation.

\bibliographystyle{utphys}
\bibliography{./QFT,./refs}

\providecommand{\href}[2]{#2}\begingroup\raggedright\begin{thebibliography}{10}

\bibitem{Nambu:1974zg}
Y.~Nambu, ``{Strings, Monopoles and Gauge Fields},''
\href{http://dx.doi.org/10.1103/PhysRevD.10.4262}{{\em Phys. Rev.} {\bfseries D10} (1974) 4262}.
%%CITATION = PHRVA,D10,4262;%%.

\bibitem{Mandelstam:1974pi}
S.~Mandelstam, \href{http://dx.doi.org/10.1016/0370-1573(76)90043-0}{``{Vortices and Quark Confinement in Nonabelian Gauge Theories},''} in {\em {Phys. Rep. 23 (1976) 245-249, In *Gervais, J.L. (Ed.), Jacob, M. (Ed.): Non-linear and Collective Phenomena In Quantum Physics*, 12-16}}, vol.~23, pp.~245--249.
\newblock
1976.
\newblock
%%CITATION = PRPLC,23,245;%%.

\bibitem{Polyakov:1975rs}
A.~M. Polyakov, ``{Compact Gauge Fields and the Infrared Catastrophe},''
\href{http://dx.doi.org/10.1016/0370-2693(75)90162-8}{{\em Phys. Lett.} {\bfseries B59} (1975) 82--84}.
%%CITATION = PHLTA,B59,82;%%.

\bibitem{tHooft:1977nqb}
G.~'t~Hooft, ``On the phase transition towards permanent quark confinement,'' \href{http://dx.doi.org/10.1016/0550-3213(78)90153-0}{{\em Nucl.Phys.B} {\bfseries 138} (1978) 1--25}.

\bibitem{Cornwall:1979hz}
J.~M. Cornwall, ``Quark confinement and vortices in massive gauge invariant {{QCD}},'' \href{http://dx.doi.org/10.1016/0550-3213(79)90111-1}{{\em Nucl. Phys. B} {\bfseries 157} no.~UCLA/79/TEP/5, (1979) 392--412}.

\bibitem{Nielsen:1979xu}
H.~B. Nielsen and P.~Olesen, ``A quantum liquid model for the {{QCD}} vacuum: {{Gauge}} and rotational invariance of domained and quantized homogeneous color fields,'' \href{http://dx.doi.org/10.1016/0550-3213(79)90065-8}{{\em Nucl. Phys. B} {\bfseries 160} no.~NBI-HE-79-17, (1979) 380--396}.

\bibitem{Ambjorn:1980ms}
J.~Ambjorn and P.~Olesen, ``A color magnetic vortex condensate in {{QCD}},'' \href{http://dx.doi.org/10.1016/0550-3213(80)90150-9}{{\em Nucl. Phys. B} {\bfseries 170} no.~NBI-HE-80-14, (1980) 265--282}.

\bibitem{Ezawa:1982bf}
Z.~F. Ezawa and A.~Iwazaki, ``{Abelian Dominance and Quark Confinement in Yang-Mills Theories},'' \href{http://dx.doi.org/10.1103/PhysRevD.25.2681}{{\em Phys. Rev. D} {\bfseries 25} (1982) 2681}.

\bibitem{Suzuki:1988yq}
T.~Suzuki, ``{A Ginzburg-Landau Type Theory of Quark Confinement},'' \href{http://dx.doi.org/10.1143/PTP.80.929}{{\em Prog. Theor. Phys.} {\bfseries 80} (1988) 929}.

\bibitem{Suganuma:1993ps}
H.~Suganuma, S.~Sasaki, and H.~Toki, ``{Color confinement, quark pair creation and dynamical chiral symmetry breaking in the dual Ginzburg-Landau theory},'' \href{http://dx.doi.org/10.1016/0550-3213(94)00392-R}{{\em Nucl. Phys. B} {\bfseries 435} (1995) 207--240}, \href{http://arxiv.org/abs/hep-ph/9312350}{{\ttfamily arXiv:hep-ph/9312350}}.

\bibitem{Kondo:1997pc}
K.-I. Kondo, ``{Abelian projected effective gauge theory of QCD with asymptotic freedom and quark confinement},'' \href{http://dx.doi.org/10.1103/PhysRevD.57.7467}{{\em Phys. Rev. D} {\bfseries 57} (1998) 7467--7487}, \href{http://arxiv.org/abs/hep-th/9709109}{{\ttfamily arXiv:hep-th/9709109}}.

\bibitem{Kronfeld:1987vd}
A.~S. Kronfeld, G.~Schierholz, and U.~J. Wiese, ``{Topology and Dynamics of the Confinement Mechanism},'' \href{http://dx.doi.org/10.1016/0550-3213(87)90080-0}{{\em Nucl. Phys. B} {\bfseries 293} (1987) 461--478}.

\bibitem{Suzuki:1989gp}
T.~Suzuki and I.~Yotsuyanagi, ``{A possible evidence for Abelian dominance in quark confinement},'' \href{http://dx.doi.org/10.1103/PhysRevD.42.4257}{{\em Phys. Rev. D} {\bfseries 42} (1990) 4257--4260}.

\bibitem{DelDebbio:1996lih}
L.~Del~Debbio, M.~Faber, J.~Greensite, and S.~Olejnik, ``Center dominance and {{Z}}(2) vortices in {{SU}}(2) lattice gauge theory,'' \href{http://dx.doi.org/10.1103/PhysRevD.55.2298}{{\em Phys. Rev. D} {\bfseries 55} no.~LBL-39424, LBNL-39424, SWAT-96-119, (1997) 2298--2306}, \href{http://arxiv.org/abs/hep-lat/9610005}{{\ttfamily arxiv:hep-lat/9610005}}.

\bibitem{Langfeld:1998cz}
K.~Langfeld, O.~Tennert, M.~Engelhardt, and H.~Reinhardt, ``Center vortices of {{Yang-Mills}} theory at finite temperatures,'' \href{http://dx.doi.org/10.1016/S0370-2693(99)00252-X}{{\em Phys. Lett. B} {\bfseries 452} no.~UNITU-THEP-5-98, (1999) 301}, \href{http://arxiv.org/abs/hep-lat/9805002}{{\ttfamily arxiv:hep-lat/9805002}}.

\bibitem{Kovacs:1998xm}
T.~G. Kovacs and E.~T. Tomboulis, ``Vortices and confinement at weak coupling,'' \href{http://dx.doi.org/10.1103/PhysRevD.57.4054}{{\em Phys. Rev. D} {\bfseries 57} no.~UCLA-97-TEP-22, COLO-HEP-392, (1998) 4054--4062}, \href{http://arxiv.org/abs/hep-lat/9711009}{{\ttfamily arxiv:hep-lat/9711009}}.

\bibitem{Engelhardt:1999fd}
M.~Engelhardt, K.~Langfeld, H.~Reinhardt, and O.~Tennert, ``Deconfinement in {{SU}}(2) {{Yang-Mills}} theory as a center vortex percolation transition,'' \href{http://dx.doi.org/10.1103/PhysRevD.61.054504}{{\em Phys. Rev. D} {\bfseries 61} (2000) 054504}, \href{http://arxiv.org/abs/hep-lat/9904004}{{\ttfamily arxiv:hep-lat/9904004}}.

\bibitem{deForcrand:1999our}
P.~{de Forcrand} and M.~D'Elia, ``On the relevance of center vortices to {{QCD}},'' \href{http://dx.doi.org/10.1103/PhysRevLett.82.4582}{{\em Phys. Rev. Lett.} {\bfseries 82} (1999) 4582--4585}, \href{http://arxiv.org/abs/hep-lat/9901020}{{\ttfamily arxiv:hep-lat/9901020}}.

\bibitem{DelDebbio:1997ke}
L.~Del~Debbio, M.~Faber, J.~Greensite, and S.~Olejnik, ``{Center dominance, center vortices, and confinement},'' in {\em {NATO Advanced Research Workshop on Theoretical Physics: New Developments in Quantum Field Theory}}, pp.~47--64.
\newblock 6, 1997.
\newblock \href{http://arxiv.org/abs/hep-lat/9708023}{{\ttfamily arXiv:hep-lat/9708023}}.

\bibitem{Ambjorn:1999ym}
J.~Ambjorn, J.~Giedt, and J.~Greensite, ``Vortex structure versus monopole dominance in {{Abelian}} projected gauge theory,'' \href{http://dx.doi.org/10.1088/1126-6708/2000/02/033}{{\em JHEP} {\bfseries 02} (2000) 033}, \href{http://arxiv.org/abs/hep-lat/9907021}{{\ttfamily arxiv:hep-lat/9907021}}.

\bibitem{deForcrand:2000pg}
P.~de~Forcrand and M.~Pepe, ``{Center vortices and monopoles without lattice Gribov copies},'' \href{http://dx.doi.org/10.1016/S0550-3213(01)00009-8}{{\em Nucl. Phys. B} {\bfseries 598} (2001) 557--577}, \href{http://arxiv.org/abs/hep-lat/0008016}{{\ttfamily arXiv:hep-lat/0008016}}.

\bibitem{Engelhardt:1999xw}
M.~Engelhardt and H.~Reinhardt, ``{Center projection vortices in continuum Yang-Mills theory},'' \href{http://dx.doi.org/10.1016/S0550-3213(99)00727-0}{{\em Nucl. Phys. B} {\bfseries 567} (2000) 249}, \href{http://arxiv.org/abs/hep-th/9907139}{{\ttfamily arXiv:hep-th/9907139}}.

\bibitem{Reinhardt:2001kf}
H.~Reinhardt, ``{Topology of center vortices},'' \href{http://dx.doi.org/10.1016/S0550-3213(02)00130-X}{{\em Nucl. Phys. B} {\bfseries 628} (2002) 133--166}, \href{http://arxiv.org/abs/hep-th/0112215}{{\ttfamily arXiv:hep-th/0112215}}.

\bibitem{Cornwall:1999xw}
J.~M. Cornwall, ``{Center vortices, nexuses, and fractional topological charge},'' \href{http://dx.doi.org/10.1103/PhysRevD.61.085012}{{\em Phys. Rev. D} {\bfseries 61} (2000) 085012}, \href{http://arxiv.org/abs/hep-th/9911125}{{\ttfamily arXiv:hep-th/9911125}}.

\bibitem{Hayashi:2024yjc}
Y.~Hayashi and Y.~Tanizaki, ``{Unifying Monopole and Center Vortex as the Semiclassical Confinement Mechanism},'' \href{http://dx.doi.org/10.1103/PhysRevLett.133.171902}{{\em Phys. Rev. Lett.} {\bfseries 133} no.~17, (2024) 171902}, \href{http://arxiv.org/abs/2405.12402}{{\ttfamily arXiv:2405.12402 [hep-th]}}.

\bibitem{GarciaPerez:1989gt}
M.~Garcia~Perez, A.~Gonzalez-Arroyo, and B.~Soderberg, ``{Minimum Action Solutions for SU(2) Gauge Theory on the Torus With Nonorthogonal Twist},'' \href{http://dx.doi.org/10.1016/0370-2693(90)90106-G}{{\em Phys. Lett. B} {\bfseries 235} (1990) 117--123}.

\bibitem{GarciaPerez:1992fj}
M.~Garcia~Perez and A.~{Gonzalez-Arroyo}, ``Numerical study of {{Yang-Mills}} classical solutions on the twisted torus,'' \href{http://dx.doi.org/10.1088/0305-4470/26/11/015}{{\em J. Phys. A} {\bfseries 26} no.~FTUAM-92-08, (1993) 2667--2678}, \href{http://arxiv.org/abs/hep-lat/9206016}{{\ttfamily arxiv:hep-lat/9206016}}.

\bibitem{Ford:2002pa}
C.~Ford and J.~M. Pawlowski, ``{Constituents of doubly periodic instantons},'' \href{http://dx.doi.org/10.1016/S0370-2693(02)02130-5}{{\em Phys. Lett. B} {\bfseries 540} (2002) 153--158}, \href{http://arxiv.org/abs/hep-th/0205116}{{\ttfamily arXiv:hep-th/0205116}}.

\bibitem{Ford:2003vi}
C.~Ford and J.~M. Pawlowski, ``{Doubly periodic instantons and their constituents},'' \href{http://dx.doi.org/10.1103/PhysRevD.69.065006}{{\em Phys. Rev. D} {\bfseries 69} (2004) 065006}, \href{http://arxiv.org/abs/hep-th/0302117}{{\ttfamily arXiv:hep-th/0302117}}.

\bibitem{Itou:2018wkm}
E.~Itou, ``Fractional instanton of the {{SU}}(3) gauge theory in weak coupling regime,'' \href{http://dx.doi.org/10.1007/JHEP05(2019)093}{{\em JHEP} {\bfseries 05} (2019) 093}, \href{http://arxiv.org/abs/1811.05708}{{\ttfamily arxiv:1811.05708 [hep-th]}}.

\bibitem{DasilvaGolan:2022jlm}
J.~Dasilva~Gol{\'a}n and M.~Garc{\'\i}a~P{\'e}rez, ``{SU(N) fractional instantons and the Fibonacci sequence},'' \href{http://dx.doi.org/10.1007/JHEP12(2022)109}{{\em JHEP} {\bfseries 12} (2022) 109}, \href{http://arxiv.org/abs/2208.07133}{{\ttfamily arXiv:2208.07133 [hep-th]}}.

\bibitem{Wandler:2024hsq}
F.~D. Wandler, ``{Numerical fractional instantons in SU(2): center vortices, monopoles, and a sharp transition between them},'' \href{http://arxiv.org/abs/2406.07636}{{\ttfamily arXiv:2406.07636 [hep-lat]}}.

\bibitem{Tanizaki:2022ngt}
Y.~Tanizaki and M.~\"Unsal, ``{Center vortex and confinement in Yang-Mills theory and QCD with anomaly-preserving compactifications},'' \href{http://dx.doi.org/10.1093/ptep/ptac042}{{\em PTEP} {\bfseries 2022} (2022) 04A108}, \href{http://arxiv.org/abs/2201.06166}{{\ttfamily arXiv:2201.06166 [hep-th]}}.

\bibitem{Tanizaki:2022plm}
Y.~Tanizaki and M.~\"Unsal, ``{Semiclassics with \textquoteright{}t Hooft flux background for QCD with 2-index quarks},'' \href{http://dx.doi.org/10.1007/JHEP08(2022)038}{{\em JHEP} {\bfseries 08} (2022) 038}, \href{http://arxiv.org/abs/2205.11339}{{\ttfamily arXiv:2205.11339 [hep-th]}}.

\bibitem{Hayashi:2023wwi}
Y.~Hayashi, Y.~Tanizaki, and H.~Watanabe, ``{Semiclassical analysis of the bifundamental QCD on~$\mathbb{R}^2\times T^2$ with \textquoteright{}t Hooft flux},'' \href{http://dx.doi.org/10.1007/JHEP10(2023)146}{{\em JHEP} {\bfseries 10} (2023) 146}, \href{http://arxiv.org/abs/2307.13954}{{\ttfamily arXiv:2307.13954 [hep-th]}}.

\bibitem{Hayashi:2024gxv}
Y.~Hayashi, Y.~Tanizaki, and H.~Watanabe, ``{Non-supersymmetric duality cascade of QCD(BF) via semiclassics on $\mathbb{R}^2\times T^2$ with the baryon-'t Hooft flux},'' \href{http://arxiv.org/abs/2404.16803}{{\ttfamily arXiv:2404.16803 [hep-th]}}.

\bibitem{Hayashi:2024qkm}
Y.~Hayashi and Y.~Tanizaki, ``{Semiclassics for the QCD vacuum structure through T$^{2}$-compactification with the baryon-{\textquoteright}t Hooft flux},'' \href{http://dx.doi.org/10.1007/JHEP08(2024)001}{{\em JHEP} {\bfseries 08} (2024) 001}, \href{http://arxiv.org/abs/2402.04320}{{\ttfamily arXiv:2402.04320 [hep-th]}}.

\bibitem{Hayashi:2024psa}
Y.~Hayashi, T.~Misumi, and Y.~Tanizaki, ``{Monopole-vortex continuity of $ \mathcal{N} $ = 1 super Yang-Mills theory on {\ensuremath{\mathbb{R}}}$^{2}$ {\texttimes} S$^{1}$ {\texttimes} S$^{1}$ with {\textquoteright}t Hooft twist},'' \href{http://dx.doi.org/10.1007/JHEP05(2025)194}{{\em JHEP} {\bfseries 05} (2025) 194}, \href{http://arxiv.org/abs/2410.21392}{{\ttfamily arXiv:2410.21392 [hep-th]}}.

\bibitem{Guvendik:2024umd}
C.~G{\"u}vendik, T.~Schaefer, and M.~{\"U}nsal, ``{The metamorphosis of semi-classical mechanisms of confinement: from monopoles on {\ensuremath{\mathbb{R}}}$^{3}$ {\texttimes} S$^{1}$ to center-vortices on {\ensuremath{\mathbb{R}}}$^{2}$ {\texttimes} T$^{2}$},'' \href{http://dx.doi.org/10.1007/JHEP11(2024)163}{{\em JHEP} {\bfseries 11} (2024) 163}, \href{http://arxiv.org/abs/2405.13696}{{\ttfamily arXiv:2405.13696 [hep-th]}}.

\bibitem{Hayashi:2025doq}
Y.~Hayashi, Y.~Tanizaki, and M.~{\"U}nsal, ``{Center-vortex semiclassics with non-minimal {\textquoteright}t Hooft fluxes on {\ensuremath{\mathbb{R}}}$^{2}${\,}{\texttimes}{\,}T$^{2}$ and center stabilization at large N},'' \href{http://dx.doi.org/10.1007/JHEP02(2026)126}{{\em JHEP} {\bfseries 02} (2026) 126}, \href{http://arxiv.org/abs/2505.07467}{{\ttfamily arXiv:2505.07467 [hep-th]}}.

\bibitem{Yamazaki:2017ulc}
M.~Yamazaki and K.~Yonekura, ``{From 4d Yang-Mills to 2d $\mathbb{CP}^{N-1}$ model: IR problem and confinement at weak coupling},'' \href{http://dx.doi.org/10.1007/JHEP07(2017)088}{{\em JHEP} {\bfseries 07} (2017) 088},
\href{http://arxiv.org/abs/1704.05852}{{\ttfamily arXiv:1704.05852 [hep-th]}}.
%%CITATION = ARXIV:1704.05852;%%.

\bibitem{Cox:2021vsa}
A.~A. Cox, E.~Poppitz, and F.~D. Wandler, ``{The mixed 0-form/1-form anomaly in Hilbert space: pouring the new wine into old bottles},'' \href{http://dx.doi.org/10.1007/JHEP10(2021)069}{{\em JHEP} {\bfseries 10} (2021) 069}, \href{http://arxiv.org/abs/2106.11442}{{\ttfamily arXiv:2106.11442 [hep-th]}}.

\bibitem{Unsal:2007vu}
M.~Unsal, ``{Abelian duality, confinement, and chiral symmetry breaking in QCD(adj)},'' \href{http://dx.doi.org/10.1103/PhysRevLett.100.032005}{{\em Phys. Rev. Lett.} {\bfseries 100} (2008) 032005},
\href{http://arxiv.org/abs/0708.1772}{{\ttfamily arXiv:0708.1772 [hep-th]}}.
%%CITATION = ARXIV:0708.1772;%%.

\bibitem{Unsal:2007jx}
M.~Unsal, ``{Magnetic bion condensation: A New mechanism of confinement and mass gap in four dimensions},'' \href{http://dx.doi.org/10.1103/PhysRevD.80.065001}{{\em Phys. Rev.} {\bfseries D80} (2009) 065001},
\href{http://arxiv.org/abs/0709.3269}{{\ttfamily arXiv:0709.3269 [hep-th]}}.
%%CITATION = ARXIV:0709.3269;%%.

\bibitem{Unsal:2008ch}
M.~Unsal and L.~G. Yaffe, ``{Center-stabilized Yang-Mills theory: Confinement and large N volume independence},'' \href{http://dx.doi.org/10.1103/PhysRevD.78.065035}{{\em Phys. Rev.} {\bfseries D78} (2008) 065035},
\href{http://arxiv.org/abs/0803.0344}{{\ttfamily arXiv:0803.0344 [hep-th]}}.
%%CITATION = ARXIV:0803.0344;%%.

\bibitem{Shifman:2008ja}
M.~Shifman and M.~Unsal, ``{QCD-like Theories on R(3) x S(1): A Smooth Journey from Small to Large r(S(1)) with Double-Trace Deformations},'' \href{http://dx.doi.org/10.1103/PhysRevD.78.065004}{{\em Phys. Rev.} {\bfseries D78} (2008) 065004},
\href{http://arxiv.org/abs/0802.1232}{{\ttfamily arXiv:0802.1232 [hep-th]}}.
%%CITATION = ARXIV:0802.1232;%%.

\bibitem{Poppitz:2012nz}
E.~Poppitz, T.~Sch\"{a}fer, and M.~\"{U}nsal, ``{Universal mechanism of (semi-classical) deconfinement and theta-dependence for all simple groups},'' \href{http://dx.doi.org/10.1007/JHEP03(2013)087}{{\em JHEP} {\bfseries 03} (2013) 087},
\href{http://arxiv.org/abs/1212.1238}{{\ttfamily arXiv:1212.1238 [hep-th]}}.
%%CITATION = ARXIV:1212.1238;%%.

\bibitem{Poppitz:2012sw}
E.~Poppitz, T.~Sch\"{a}fer, and M.~\"{U}nsal, ``{Continuity, Deconfinement, and (Super) Yang-Mills Theory},'' \href{http://dx.doi.org/10.1007/JHEP10(2012)115}{{\em JHEP} {\bfseries 10} (2012) 115},
\href{http://arxiv.org/abs/1205.0290}{{\ttfamily arXiv:1205.0290 [hep-th]}}.
%%CITATION = ARXIV:1205.0290;%%.

\bibitem{Davies:2000nw}
N.~M. Davies, T.~J. Hollowood, and V.~V. Khoze, ``{Monopoles, affine algebras and the gluino condensate},'' \href{http://dx.doi.org/10.1063/1.1586477}{{\em J. Math. Phys.} {\bfseries 44} (2003) 3640--3656},
\href{http://arxiv.org/abs/hep-th/0006011}{{\ttfamily arXiv:hep-th/0006011 [hep-th]}}.
%%CITATION = HEP-TH/0006011;%%.

\bibitem{Davies:1999uw}
N.~M. Davies, T.~J. Hollowood, V.~V. Khoze, and M.~P. Mattis, ``{Gluino condensate and magnetic monopoles in supersymmetric gluodynamics},'' \href{http://dx.doi.org/10.1016/S0550-3213(99)00434-4}{{\em Nucl. Phys. B} {\bfseries 559} (1999) 123--142}, \href{http://arxiv.org/abs/hep-th/9905015}{{\ttfamily arXiv:hep-th/9905015}}.

\bibitem{Nguyen:2024ikq}
M.~Nguyen, T.~Sulejmanpasic, and M.~{\"U}nsal, ``{Phases of Theories with ZN 1-Form Symmetry, and the Roles of Center Vortices and Magnetic Monopoles},'' \href{http://dx.doi.org/10.1103/PhysRevLett.134.141902}{{\em Phys. Rev. Lett.} {\bfseries 134} no.~14, (2025) 141902}, \href{http://arxiv.org/abs/2401.04800}{{\ttfamily arXiv:2401.04800 [hep-th]}}.

\bibitem{Giansiracusa:2025hfj}
J.~Giansiracusa, D.~Lanners, and T.~Sulejmanpasic, ``{Emergent Photons and Confinement: A Numerical Study on ZN Lattice Gauge Theory},'' \href{http://dx.doi.org/10.1103/h8mn-t4fk}{{\em Phys. Rev. Lett.} {\bfseries 135} no.~22, (2025) 221901}, \href{http://arxiv.org/abs/2505.00079}{{\ttfamily arXiv:2505.00079 [hep-lat]}}.

\bibitem{Gaiotto:2014kfa}
D.~Gaiotto, A.~Kapustin, N.~Seiberg, and B.~Willett, ``{Generalized Global Symmetries},'' \href{http://dx.doi.org/10.1007/JHEP02(2015)172}{{\em JHEP} {\bfseries 02} (2015) 172},
\href{http://arxiv.org/abs/1412.5148}{{\ttfamily arXiv:1412.5148 [hep-th]}}.
%%CITATION = ARXIV:1412.5148;%%.

\bibitem{Kapustin:2014gua}
A.~Kapustin and N.~Seiberg, ``{Coupling a QFT to a TQFT and Duality},'' \href{http://dx.doi.org/10.1007/JHEP04(2014)001}{{\em JHEP} {\bfseries 04} (2014) 001},
\href{http://arxiv.org/abs/1401.0740}{{\ttfamily arXiv:1401.0740 [hep-th]}}.
%%CITATION = ARXIV:1401.0740;%%.

\bibitem{Sharpe:2015mja}
E.~Sharpe, ``{Notes on generalized global symmetries in QFT},'' \href{http://dx.doi.org/10.1002/prop.201500048}{{\em Fortsch. Phys.} {\bfseries 63} (2015) 659--682},
\href{http://arxiv.org/abs/1508.04770}{{\ttfamily arXiv:1508.04770 [hep-th]}}.
%%CITATION = ARXIV:1508.04770;%%.

\bibitem{Witten:1988hf}
E.~Witten, ``{Quantum Field Theory and the Jones Polynomial},''
\href{http://dx.doi.org/10.1007/BF01217730}{{\em Commun.Math.Phys.} {\bfseries 121} (1989) 351--399}.
%%CITATION = CMPHA,121,351;%%.

\bibitem{Atiyah:1989vu}
M.~Atiyah, ``{Topological quantum field theories},'' \href{http://dx.doi.org/10.1007/BF02698547}{{\em Inst. Hautes Etudes Sci. Publ. Math.} {\bfseries 68} (1989) 175--186}.

\bibitem{Baez:1995xq}
J.~C. Baez and J.~Dolan, ``{Higher dimensional algebra and topological quantum field theory},'' \href{http://dx.doi.org/10.1063/1.531236}{{\em J. Math. Phys.} {\bfseries 36} (1995) 6073--6105}, \href{http://arxiv.org/abs/q-alg/9503002}{{\ttfamily arXiv:q-alg/9503002}}.

\bibitem{Lurie:2009keu}
J.~Lurie, ``{On the Classification of Topological Field Theories},'' \href{http://arxiv.org/abs/0905.0465}{{\ttfamily arXiv:0905.0465 [math.CT]}}.

\bibitem{Kapustin:2010ta}
A.~Kapustin, ``{Topological Field Theory, Higher Categories, and Their Applications},''
\href{http://arxiv.org/abs/1004.2307}{{\ttfamily arXiv:1004.2307 [math.QA]}}.
%%CITATION = ARXIV:1004.2307;%%.

\bibitem{Freed:2016rqq}
D.~S. Freed and M.~J. Hopkins, ``{Reflection positivity and invertible topological phases},'' \href{http://dx.doi.org/10.2140/gt.2021.25.1165}{{\em Geom. Topol.} {\bfseries 25} (2021) 1165--1330}, \href{http://arxiv.org/abs/1604.06527}{{\ttfamily arXiv:1604.06527 [hep-th]}}.

\bibitem{Lan:2018vjb}
T.~Lan, L.~Kong, and X.-G. Wen, ``{Classification of (3+1)D Bosonic Topological Orders: The Case When Pointlike Excitations Are All Bosons},'' \href{http://dx.doi.org/10.1103/PhysRevX.8.021074}{{\em Phys. Rev. X} {\bfseries 8} no.~2, (2018) 021074}.

\bibitem{Lan:2018bui}
T.~Lan and X.-G. Wen, ``{Classification of 3+1D Bosonic Topological Orders (II): The Case When Some Pointlike Excitations Are Fermions},'' \href{http://dx.doi.org/10.1103/PhysRevX.9.021005}{{\em Phys. Rev. X} {\bfseries 9} no.~2, (2019) 021005}, \href{http://arxiv.org/abs/1801.08530}{{\ttfamily arXiv:1801.08530 [cond-mat.str-el]}}.

\bibitem{Kitaev:2006lla}
A.~Kitaev, ``{Anyons in an exactly solved model and beyond},'' \href{http://dx.doi.org/10.1016/j.aop.2005.10.005}{{\em Annals Phys.} {\bfseries 321} no.~1, (2006) 2--111},
\href{http://arxiv.org/abs/cond-mat/0506438}{{\ttfamily arXiv:cond-mat/0506438 [cond-mat.mes-hall]}}.
%%CITATION = ARXIV:cond-mat/0506438;%%.

\bibitem{Johnson-Freyd:2020usu}
T.~Johnson-Freyd, ``{On the Classification of Topological Orders},'' \href{http://dx.doi.org/10.1007/s00220-022-04380-3}{{\em Commun. Math. Phys.} {\bfseries 393} no.~2, (2022) 989--1033}, \href{http://arxiv.org/abs/2003.06663}{{\ttfamily arXiv:2003.06663 [math.CT]}}.

\bibitem{Gu:2009dr}
Z.-C. Gu and X.-G. Wen, ``{Tensor-Entanglement-Filtering Renormalization Approach and Symmetry Protected Topological Order},'' \href{http://dx.doi.org/10.1103/PhysRevB.80.155131}{{\em Phys. Rev. B} {\bfseries 80} (2009) 155131}, \href{http://arxiv.org/abs/0903.1069}{{\ttfamily arXiv:0903.1069 [cond-mat.str-el]}}.

\bibitem{Maeda:2025ycr}
J.~Maeda and Y.~Tanizaki, ``{Twisted partition functions as order parameters},'' \href{http://dx.doi.org/10.1007/JHEP08(2025)128}{{\em JHEP} {\bfseries 08} (2025) 128}, \href{http://arxiv.org/abs/2505.16546}{{\ttfamily arXiv:2505.16546 [hep-th]}}.

\bibitem{tHooft:1979rtg}
G.~'t~Hooft, ``{A Property of Electric and Magnetic Flux in Nonabelian Gauge Theories},''
\href{http://dx.doi.org/10.1016/0550-3213(79)90595-9}{{\em Nucl. Phys.} {\bfseries B153} (1979) 141--160}.
%%CITATION = NUPHA,B153,141;%%.

\bibitem{Nguyen:2023fun}
M.~Nguyen, Y.~Tanizaki, and M.~{\"U}nsal, ``{Study of gapped phases of 4d gauge theories using temporal gauging of the~$\mathbb{Z}_N$ 1-form symmetry},'' \href{http://dx.doi.org/10.1007/JHEP08(2023)013}{{\em JHEP} {\bfseries 08} (2023) 013}, \href{http://arxiv.org/abs/2306.02485}{{\ttfamily arXiv:2306.02485 [hep-th]}}.

\bibitem{tHooft:1974kcl}
G.~'t~Hooft, ``{Magnetic Monopoles in Unified Gauge Theories},''
\href{http://dx.doi.org/10.1016/0550-3213(74)90486-6}{{\em Nucl. Phys.} {\bfseries B79} (1974) 276--284}.
%%CITATION = NUPHA,B79,276;%%.

\bibitem{Polyakov:1974ek}
A.~M. Polyakov, ``{Particle Spectrum in the Quantum Field Theory},''
{\em JETP Lett.} {\bfseries 20} (1974) 194--195.
%%CITATION = JTPLA,20,194;%%.

\bibitem{Kapustin:2013uxa}
A.~Kapustin and R.~Thorngren, ``{Higher symmetry and gapped phases of gauge theories},'' \href{http://arxiv.org/abs/1309.4721}{{\ttfamily arXiv:1309.4721 [hep-th]}}.

\bibitem{Thorngren:2020aph}
R.~Thorngren, ``{Topological quantum field theory, symmetry breaking, and finite gauge theory in 3+1D},'' \href{http://dx.doi.org/10.1103/PhysRevB.101.245160}{{\em Phys. Rev. B} {\bfseries 101} no.~24, (2020) 245160}, \href{http://arxiv.org/abs/2001.11938}{{\ttfamily arXiv:2001.11938 [cond-mat.str-el]}}.

\bibitem{Chen:2014wse}
X.~Chen, F.~J. Burnell, A.~Vishwanath, and L.~Fidkowski, ``{Anomalous Symmetry Fractionalization and Surface Topological Order},'' \href{http://dx.doi.org/10.1103/PhysRevX.5.041013}{{\em Phys. Rev. X} {\bfseries 5} no.~4, (2015) 041013}, \href{http://arxiv.org/abs/1403.6491}{{\ttfamily arXiv:1403.6491 [cond-mat.str-el]}}.

\bibitem{Barkeshli:2014cna}
M.~Barkeshli, P.~Bonderson, M.~Cheng, and Z.~Wang, ``{Symmetry Fractionalization, Defects, and Gauging of Topological Phases},'' \href{http://dx.doi.org/10.1103/PhysRevB.100.115147}{{\em Phys. Rev. B} {\bfseries 100} no.~11, (2019) 115147}, \href{http://arxiv.org/abs/1410.4540}{{\ttfamily arXiv:1410.4540 [cond-mat.str-el]}}.

\bibitem{Hsin:2019fhf}
P.-S. Hsin and A.~Turzillo, ``{Symmetry-enriched quantum spin liquids in (3 + 1)$d$},'' \href{http://dx.doi.org/10.1007/JHEP09(2020)022}{{\em JHEP} {\bfseries 09} (2020) 022}, \href{http://arxiv.org/abs/1904.11550}{{\ttfamily arXiv:1904.11550 [cond-mat.str-el]}}.

\bibitem{Delmastro:2022pfo}
D.~G. Delmastro, J.~Gomis, P.-S. Hsin, and Z.~Komargodski, ``{Anomalies and symmetry fractionalization},'' \href{http://dx.doi.org/10.21468/SciPostPhys.15.3.079}{{\em SciPost Phys.} {\bfseries 15} no.~3, (2023) 079}, \href{http://arxiv.org/abs/2206.15118}{{\ttfamily arXiv:2206.15118 [hep-th]}}.

\bibitem{Hsin:2024aqb}
P.-S. Hsin, R.~Kobayashi, and C.~Zhang, ``{Fractionalization of coset non-invertible symmetry and exotic Hall conductance},'' \href{http://dx.doi.org/10.21468/SciPostPhys.17.3.095}{{\em SciPost Phys.} {\bfseries 17} no.~3, (2024) 095}, \href{http://arxiv.org/abs/2405.20401}{{\ttfamily arXiv:2405.20401 [cond-mat.str-el]}}.

\bibitem{Brennan:2025acl}
T.~D. Brennan, T.~Jacobson, and K.~Roumpedakis, ``{Consequences of symmetry fractionalization without 1-form global symmetries},'' \href{http://dx.doi.org/10.1007/JHEP11(2025)153}{{\em JHEP} {\bfseries 11} (2025) 153}, \href{http://arxiv.org/abs/2504.08036}{{\ttfamily arXiv:2504.08036 [hep-th]}}.

\bibitem{Engelhardt:1999wr}
M.~Engelhardt and H.~Reinhardt, ``{Center vortex model for the infrared sector of Yang-Mills theory: Confinement and deconfinement},'' \href{http://dx.doi.org/10.1016/S0550-3213(00)00445-4}{{\em Nucl. Phys. B} {\bfseries 585} (2000) 591--613}, \href{http://arxiv.org/abs/hep-lat/9912003}{{\ttfamily arXiv:hep-lat/9912003}}.

\bibitem{Oxman:2018dzp}
L.~E. Oxman, ``{4D ensembles of percolating center vortices and monopole defects: The emergence of flux tubes with N -ality and gluon confinement},'' \href{http://dx.doi.org/10.1103/PhysRevD.98.036018}{{\em Phys. Rev. D} {\bfseries 98} no.~3, (2018) 036018}, \href{http://arxiv.org/abs/1805.06354}{{\ttfamily arXiv:1805.06354 [hep-th]}}.

\bibitem{Junior:2019fty}
D.~R. Junior, L.~E. Oxman, and G.~M. Sim\~oes, ``{3D Yang-Mills confining properties from a non-Abelian ensemble perspective},'' \href{http://dx.doi.org/10.1007/JHEP01(2020)180}{{\em JHEP} {\bfseries 01} (2020) 180}, \href{http://arxiv.org/abs/1911.10144}{{\ttfamily arXiv:1911.10144 [hep-th]}}.

\bibitem{Junior:2022bol}
D.~R. Junior, L.~E. Oxman, and H.~Reinhardt, ``{Infrared Yang-Mills wave functional due to percolating center vortices},'' \href{http://dx.doi.org/10.1103/PhysRevD.106.114021}{{\em Phys. Rev. D} {\bfseries 106} no.~11, (2022) 114021}, \href{http://arxiv.org/abs/2211.03006}{{\ttfamily arXiv:2211.03006 [hep-th]}}.

\bibitem{tHooft:1981bkw}
G.~'t~Hooft, ``{Topology of the Gauge Condition and New Confinement Phases in Nonabelian Gauge Theories},''
\href{http://dx.doi.org/10.1016/0550-3213(81)90442-9}{{\em Nucl. Phys.} {\bfseries B190} (1981) 455--478}.
%%CITATION = NUPHA,B190,455;%%.

\bibitem{Cardy:1981qy}
J.~L. Cardy and E.~Rabinovici, ``{Phase Structure of Z(p) Models in the Presence of a Theta Parameter},''
\href{http://dx.doi.org/10.1016/0550-3213(82)90463-1}{{\em Nucl. Phys.} {\bfseries B205} (1982) 1--16}.
%%CITATION = NUPHA,B205,1;%%.

\bibitem{Cardy:1981fd}
J.~L. Cardy, ``{Duality and the Theta Parameter in Abelian Lattice Models},''
\href{http://dx.doi.org/10.1016/0550-3213(82)90464-3}{{\em Nucl. Phys.} {\bfseries B205} (1982) 17--26}.
%%CITATION = NUPHA,B205,17;%%.

\bibitem{Honda:2020txe}
M.~Honda and Y.~Tanizaki, ``{Topological aspects of $4$D Abelian lattice gauge theories with the $\theta$ parameter},'' \href{http://dx.doi.org/10.1007/JHEP12(2020)154}{{\em JHEP} {\bfseries 12} (2020) 154}, \href{http://arxiv.org/abs/2009.10183}{{\ttfamily arXiv:2009.10183 [hep-th]}}.

\bibitem{Hayashi:2022fkw}
Y.~Hayashi and Y.~Tanizaki, ``{Non-invertible self-duality defects of Cardy-Rabinovici model and mixed gravitational anomaly},'' \href{http://dx.doi.org/10.1007/JHEP08(2022)036}{{\em JHEP} {\bfseries 08} (2022) 036}, \href{http://arxiv.org/abs/2204.07440}{{\ttfamily arXiv:2204.07440 [hep-th]}}.

\bibitem{Gaiotto:2017yup}
D.~Gaiotto, A.~Kapustin, Z.~Komargodski, and N.~Seiberg, ``{Theta, Time Reversal, and Temperature},'' \href{http://dx.doi.org/10.1007/JHEP05(2017)091}{{\em JHEP} {\bfseries 05} (2017) 091},
\href{http://arxiv.org/abs/1703.00501}{{\ttfamily arXiv:1703.00501 [hep-th]}}.
%%CITATION = ARXIV:1703.00501;%%.

\bibitem{Witten:1980sp}
E.~Witten, ``{Large N Chiral Dynamics},''
\href{http://dx.doi.org/10.1016/0003-4916(80)90325-5}{{\em Annals Phys.} {\bfseries 128} (1980) 363}.
%%CITATION = APNYA,128,363;%%.

\bibitem{DiVecchia:1980yfw}
P.~Di~Vecchia and G.~Veneziano, ``{Chiral Dynamics in the Large n Limit},'' \href{http://dx.doi.org/10.1016/0550-3213(80)90370-3}{{\em Nucl. Phys. B} {\bfseries 171} (1980) 253--272}.

\bibitem{Gagliano:2025gwr}
F.~Gagliano, A.~Grigoletto, and K.~Ohmori, ``{Higher Representations and Quark Confinement},'' \href{http://arxiv.org/abs/2501.09069}{{\ttfamily arXiv:2501.09069 [hep-th]}}.

\bibitem{Cordova:2024iti}
C.~Cordova, N.~Holfester, and K.~Ohmori, ``{Representation theory of solitons},'' \href{http://dx.doi.org/10.1007/JHEP06(2025)001}{{\em JHEP} {\bfseries 06} (2025) 001}, \href{http://arxiv.org/abs/2408.11045}{{\ttfamily arXiv:2408.11045 [hep-th]}}.

\bibitem{Benini:2011nc}
F.~Benini, T.~Nishioka, and M.~Yamazaki, ``{4d Index to 3d Index and 2d TQFT},'' \href{http://dx.doi.org/10.1103/PhysRevD.86.065015}{{\em Phys. Rev. D} {\bfseries 86} (2012) 065015}, \href{http://arxiv.org/abs/1109.0283}{{\ttfamily arXiv:1109.0283 [hep-th]}}.

\bibitem{Razamat:2013jxa}
S.~S. Razamat and M.~Yamazaki, ``{S-duality and the N=2 Lens Space Index},'' \href{http://dx.doi.org/10.1007/JHEP10(2013)048}{{\em JHEP} {\bfseries 10} (2013) 048}, \href{http://arxiv.org/abs/1306.1543}{{\ttfamily arXiv:1306.1543 [hep-th]}}.

\bibitem{Razamat:2013opa}
S.~S. Razamat and B.~Willett, ``{Global Properties of Supersymmetric Theories and the Lens Space},'' \href{http://dx.doi.org/10.1007/s00220-014-2111-0}{{\em Commun. Math. Phys.} {\bfseries 334} no.~2, (2015) 661--696}, \href{http://arxiv.org/abs/1307.4381}{{\ttfamily arXiv:1307.4381 [hep-th]}}.

\bibitem{Anber:2026pfg}
M.~M. Anber, ``{$S$-duality, boundary states, and higher-form symmetries on ALE spaces},'' \href{http://arxiv.org/abs/2605.26224}{{\ttfamily arXiv:2605.26224 [hep-th]}}.

\bibitem{Bianchi:1996zj}
M.~Bianchi, F.~Fucito, G.~Rossi, and M.~Martellini, ``{Explicit construction of Yang-Mills instantons on ALE spaces},'' \href{http://dx.doi.org/10.1016/0550-3213(96)00240-4}{{\em Nucl. Phys. B} {\bfseries 473} (1996) 367--404}, \href{http://arxiv.org/abs/hep-th/9601162}{{\ttfamily arXiv:hep-th/9601162}}.

\bibitem{Hatcher:0}
A.~Hatcher, {\em {Algebraic topology}}.
\newblock Cambridge Univ. Press, Cambridge, 2000.
\newblock \url{https://pi.math.cornell.edu/~hatcher/AT/ATpage.html}.

\end{thebibliography}\endgroup

\end{document}